\documentclass[journal,draftclsnofoot,onecolumn,12pt]{IEEEtran}

\newif\ifarxiv
\arxivtrue % set to true for arXiv version, false for conference version

\newif\ifonecolumn
\onecolumntrue % Uncomment this line when using one column format
\ifonecolumn
\else
 
\fi
\IEEEoverridecommandlockouts
\usepackage{multirow}
\usepackage{dsfont}
\usepackage{cite}
\usepackage{amsmath,amssymb,amsfonts, mathtools}
\usepackage{float}
\usepackage{makecell}
\usepackage{caption}
\usepackage{graphicx}
\usepackage{textcomp}
\usepackage{algorithm}
\usepackage[left=0.625in,right=0.625in,bottom=0.95in,top=0.75in]{geometry}
\usepackage{amsthm}
\usepackage{algpseudocode}
\usepackage{adjustbox}
\usepackage{verbatim}
\usepackage{paralist}
\usepackage{array}
\newcolumntype{?}{!{\vrule width 1.1pt}}
\usepackage{tablefootnote}
\usepackage[table,xcdraw]{xcolor}

\DeclareMathOperator*{\argmax}{arg\,max}

\newtheorem{proposition}[]{Proposition}

\theoremstyle{remark}

\usepackage{xcolor} 
\definecolor{green}{rgb}{0.0, 0.5, 0.0} % This is a dark shade of green
\allowdisplaybreaks
\usepackage{acro}
\newcolumntype{?}{!{\vrule width 1pt}}

\DeclareAcronym{snr}{
  short = SNR,
  long = signal-to-noise ratio,}
  \DeclareAcronym{pdf}{
	short = PDF,
	long = probability density function,}
  \DeclareAcronym{sar}{
  short = SAR,
  long = synthetic aperture radar,}

\DeclareAcronym{insar}{
  short = InSAR,
  long = interferometric SAR,}

\DeclareAcronym{gmti}{
	short = GMTI,
	long =ground moving target indication,}
	
\DeclareAcronym{ao}{
	short = {AO},
	long = {alternating optimization},
	long-plural-form = {alternating optimizations}
}
\DeclareAcronym{mimo}{
	short = {MIMO},
	long = {multiple-input multiple-output}
}

\DeclareAcronym{uav}{
        short = {UAV},
        long = {unmanned aerial vehicle},
        long-plural-form = {unmanned aerial vehicles}
}

\DeclareAcronym{fdma}{
	short = {FDMA},
	long = {frequency-division multiple-access},
}
\DeclareAcronym{1d}{
	short = {1D},
	long = {one-dimensional},
}
\DeclareAcronym{islr}{
	short = {ISLR},
	long = {integrated sidelobe ratio},
}

\DeclareAcronym{pslr}{
	short = {PSLR},
	long = {peak sidelobe ratio},
}

\DeclareAcronym{3d}{
        short = {3D},
        long = {three-dimensional},
}
\DeclareAcronym{pso}{
	short = {PSO},
	long = {particle swarm optimization},
}

\DeclareAcronym{2d}{
        short = {2D},
        long = {two-dimensional},
}

\DeclareAcronym{dem}{
        short = {DEM},
        long = {digital elevation model},
}
\DeclareAcronym{gs}{
        short = {GS},
        long = {ground station},
        long-plural-form = {ground stations}
}

\DeclareAcronym{los}{
        short = {LOS},
        long = {line-of-sight},
}

\DeclareAcronym{sca}{
        short = {SCA},
        long = {successive convex approximation},
}

\DeclareAcronym{nesz}{
        short = {NESZ},
        long = {noise equivalent sigma zero},
}

\DeclareAcronym{wrt}{
        short = {w.r.t.},
        long = {with respect to },
}
\DeclareAcronym{rhs}{
        short = {r.h.s},
        long = {right-hand side},
        long-plural-form ={right-hand sides},
}

\DeclareAcronym{lhs}{
        short = {l.h.s},
        long = {left-hand side },
}
\DeclareAcronym{bcd}{
	short = {BCD},
	long = {block coordinate descent},
}
\DeclareAcronym{hoa}{
        short = {HoA},
        long = {height of ambiguity},
}

\definecolor{maxvalue}{RGB}{181, 243, 167}
\definecolor{midvalue1}{RGB}{245, 254, 185}
\definecolor{midvalue2}{RGB}{254, 220, 185 }
\definecolor{midvalue3}{RGB}{254, 207, 185 }
\definecolor{minvalue}{RGB}{249, 181, 181 }

\begin{document}
\title{UAV Formation and Resource {Allocation} Optimization for Communication-{Assisted} 3D  InSAR Sensing\\
\thanks{This paper {is} presented in part at the IEEE International Conference on Communication, Denver, USA, Jun. 2024  \cite{conf2}. This work was supported in part by the Deutsche Forschungsgemeinschaft (DFG, German Research Foundation) GRK 2680 – Project-ID 437847244.}
}
\author{\IEEEauthorblockN{Mohamed-Amine~Lahmeri\IEEEauthorrefmark{1}, Víctor Mustieles-Pérez\IEEEauthorrefmark{1}\IEEEauthorrefmark{2}, Martin Vossiek\IEEEauthorrefmark{1}, Gerhard Krieger\IEEEauthorrefmark{1}\IEEEauthorrefmark{2}, and
Robert Schober\IEEEauthorrefmark{1}}\\
\IEEEauthorblockA{\IEEEauthorrefmark{1}Friedrich-Alexander-Universit\"at Erlangen-N\"urnberg {(FAU)}, Germany\\
\IEEEauthorrefmark{2}German Aerospace Center (DLR),  Microwaves and Radar Institute, Weßling, Germany\\
\vspace{-8mm}}}
\maketitle
\begin{abstract} 
	\sloppy
	In this paper, we investigate joint unmanned aerial vehicle (UAV) formation and resource allocation optimization for communication-assisted three-dimensional (3D) synthetic aperture radar (SAR) sensing. We consider a system consisting of two UAVs that perform bistatic interferometric SAR (InSAR) sensing {for generation of}  a digital elevation model (DEM) and transmit the radar raw data to a ground station (GS) in {real time}. To account for practical 3D sensing requirements, we use non-conventional sensing performance metrics, such as the interferometric coherence, i.e., the local cross-correlation between the two co-registered UAV SAR images, the   point-to-point InSAR relative height error, and the height of ambiguity, which together characterize the accuracy with which the InSAR system can  {determine} {the height of ground targets}. Our objective is to jointly optimize the UAV formation, speed, and communication power allocation {for maximization of} the InSAR coverage while satisfying energy, communication, and InSAR-specific sensing constraints. To solve the formulated non-smooth and non-convex optimization problem, we  divide {it} into three sub-problems and propose a novel alternating optimization (AO) framework that is based on classical, monotonic, and stochastic optimization techniques. The effectiveness of the proposed algorithm is validated through extensive simulations and compared to several benchmark schemes. {Furthermore, our simulation results} highlight the impact of the UAV-GS  communication link on the flying formation and sensing performance and  show that the DEM of a large area of interest can be mapped and offloaded to ground successfully, while the ground {topography} can be estimated with centimeter-scale precision. Lastly, we demonstrate that a low UAV velocity is preferable for InSAR applications as it leads to better sensing accuracy.
\end{abstract}

\section{Introduction}
The widespread use of \acp{uav} has revolutionized modern technology, impacting fields like remote sensing, communication, and disaster  monitoring  \cite{uav_survey}. Specifically, \acp{uav} excel in remote sensing applications using a variety of technologies, such as cameras, lidars, and radars. Among these technologies, radar sensing has attracted particular attention due to its capacity to operate day and night under all weather conditions. {Radar sensing applications can be classified into three distinct categories based on the dimensionality of the extracted information}. Firstly, \ac{1d} sensing includes applications such as ranging and detection and necessitates at least one radar antenna. Secondly, \ac{2d} sensing, also known as radar imaging, can be realized using, e.g., \ac{mimo} radar systems employing multiple antennas \cite{mimo} or by  moving a {single} radar antenna {or a \ac{mimo} array} along a trajectory to create a \ac{sar} {or \ac{mimo}-\ac{sar}} system, and can provide a {high-resolution} image of the illuminated area \cite{tutorial}. {Lastly, \ac{3d} sensing provides an additional dimension besides the {\ac{2d}} spatial location of a target, such as its altitude above ground. Here, an interesting remote sensing technique that can be used to generate accurate \ac{3d} \acp{dem} is across-track \ac{insar} \cite{tutorial}. In the following subsection, we discuss \ac{1d} and \ac{2d} \ac{uav}-based sensing applications, focusing on the existing trajectory and resource optimization frameworks. Then, in Subsection \ref{Sec:Introduction_B}, we address \ac{uav}-based \ac{3d} sensing applications, focusing on \ac{insar} applications, which are the main topic of interest in this paper.}\par
{\subsection{Existing  Optimization Frameworks for \ac{uav}-based \ac{1d} and \ac{2d} Sensing}}
\label{Sec:Introduction_A}

{For \ac{1d} sensing, the joint optimization of the trajectory and beamforming of \ac{uav}-based integrated sensing and communication systems has been studied to minimize the average power consumption \cite{khalili} and to maximize the achievable communication rate for both single-user \cite{beampattern2} and multi-user \cite{beampattern} scenarios. For \ac{2d} sensing, trajectory optimization for \ac{uav}-based \ac{sar} was investigated in \cite{response1,sun1,sun2,sun3,2D_sensing,2D_sensing_BS,amine1,amine2}. In \cite{sun1,sun2,sun3}, the backscattered signal was received  by  a \ac{uav}, while the target area was illuminated by a geosynchronous satellite. In \cite{response1,2D_sensing,amine1,amine2}, \acp{uav} were used as transmitters and receivers and, in \cite{2D_sensing_BS}, a \ac{uav} was used as receiver, while the target area was illuminated by a communication base station.} For these systems, seamless \ac{uav}-to-ground connectivity enables the timely  transfer of essential information to {the} ground for processing. In fact, the authors in \cite{amine1,amine2,sun3} assumed a \ac{los} link between the \ac{uav}-\ac{sar} system and a \ac{gs}, and, in \cite{amine1,amine2}, real-time offloading to the \ac{gs} was considered, while in \cite{sun3}, the authors investigated transmitting the total radar raw data to {the} ground. {In general, trajectory and resource allocation optimization for  \ac{uav}-based sensing  has been extensively studied in the literature. However, most works focus on \ac{1d} or \ac{2d} sensing applications. Although a  performance analysis for \ac{3d} sensing applications {was} provided in \cite{3D_sensing,snr_equation,tradeoff}, the joint {\ac{uav} formation} and resource allocation {optimization} for such applications has not been studied yet.}\par
{\subsection{\ac{uav}-based \ac{3d} Sensing via \ac{insar}\label{Sec:Introduction_B}}
Across-track \ac{insar} is an interesting \ac{3d} remote sensing technique that requires the use of more than one radar antenna. At least two radar antennas are moved to  illuminate a given target area from different angles, such that at least two \ac{sar} images can be formed.} All possible image pairs are first co-registered \cite{insar_introduction}, i.e., aligned such that each ground point corresponds to the same pixel in both {images}, and based on the phase difference, an interferogram is created to extract information about the topography of the target area  \cite{insar_introduction,coherence1}.
\ac{insar}-based \ac{3d} sensing is fundamentally different from \ac{1d} and \ac{2d} sensing. In particular, to assess the performance of \ac{1d} sensing, conventional metrics, such as \ac{snr} \cite{khalili}, beam pattern gain \cite{beampattern,beampattern2}, and probability of detection and false alarm \cite{cooperative,false_alarm_rate} can be {adopted}. For \ac{2d} sensing, image performance metrics can be  used, including \ac{sar} \ac{snr} \cite{amine1}, spatial resolution \cite{sun1}, and \ac{islr} and \ac{pslr} \cite{pslr}. However, these performance metrics are not sufficient for \ac{3d} sensing applications, such as generating an accurate \ac{dem} based on across-track interferometry. In fact, the key performance metric for estimating interferometric performance  is  coherence, which is a function of the correlation between {the} pair of  co-registered \ac{sar} images \cite{coherence1}. Additional relevant performance metrics are the \ac{hoa}, which is a proportionality constant between the interferometric phase and the terrain height and is thus related to the sensitivity of the radar to the ground topography \cite{coherence1}, and the relative height accuracy, which represents the accuracy with which the altitude of ground objects {(i.e., the ground topography)} can be  estimated.  {In this context, an interesting trade-off arises; while a large inter-\ac{uav} separation distance, i.e., a large interferometric baseline, leads to  a small \ac{hoa}, which in general improves the height accuracy of  \acp{dem}, it also leads to a degradation of the  image coherence \cite{tradeoff}. Furthermore, \ac{3d}  {\ac{insar}} sensing is different from cooperative sensing \cite{cooperative,collaborative_sensing_communication} {for} two main reasons. In particular, for \ac{3d} \ac{insar} (i) the coverage  is defined by the region where the ground footprints of the two radar antennas overlap and (ii) the sensing performance is highly { dependent on} the \ac{uav} formation. Moreover, compared to \ac{1d} and \ac{2d} sensing, \ac{3d} sensing is more challenging due to the use {of} multiple sensors and the need for {considering} the aforementioned dedicated performance metrics. Specifically, relying solely on beam pattern design, \ac{snr}, or resolution is not sufficient for \ac{3d} applications. Therefore, based on the above discussion, existing results for \ac{uav}-based \ac{1d} and \ac{2d}  sensing are not applicable to interferometric \ac{3d} sensing.}	\par
{\subsection{Main Contributions}}In this paper, we consider a system of two \acp{uav} that perform \ac{sar} sensing of a given area from different angles, such that the difference of the sensors' {phases} allows for estimating the height of ground objects based on \ac{insar} techniques. For this system, we maximize the sensing coverage  by jointly optimizing the {formation, velocity,} and communication transmit power of both \acp{uav} while satisfying practical energy, communication, and sensing constraints. Our main contributions can be summarized as follows:
\begin{itemize}
\item  {Different from existing works targeting \ac{1d} and \ac{2d} sensing applications, we consider \ac{insar}-specific sensing performance metrics, such as \ac{insar} coverage, coherence, \ac{hoa}, and relative height accuracy, that can capture the quality of \acp{dem}.}
\item {{A} joint \ac{uav} formation and resource allocation optimization problem is formulated to maximize the sensing coverage while satisfying practical energy, communication, and sensing constraints.}
\item A novel \ac{ao}-based algorithm, which incorporates classical, monotonic, and stochastic optimization { techniques}, is proposed to solve the  formulated challenging {non-smooth and non-convex} optimization problem.
\item { Our simulations show the superior performance of the proposed solution compared to {a} conventional \ac{ao} algorithm and other benchmark schemes and unveil the  impact of the {\ac{insar}-specific} sensing requirements on communication and vice versa.}
\end{itemize}
We note that this article extends a corresponding conference version \cite{conf2}. In {contrast to this paper, in} \cite{conf2}, the velocity of both the master and slave \acp{uav} was fixed and the radar look angle of the slave drone was {also} not optimized. Furthermore, the minimum sensing data rate was assumed to be fixed, i.e., independent of the imageable area. Lastly, the relative point-to-point height accuracy was not accounted for in \cite{conf2}.\par {The remainder of this paper is organized as follows. In Section II, we present the considered system. In Section III, the formation and resource allocation problem is formulated, and the proposed solution is presented in  Section IV. {In Section V,} the performance of the proposed {scheme} is assessed via simulations, and {finally,} conclusions are drawn in {Section VI}.}\par
{\em Notations}: 
In this paper, lower-case letters $x$ refer to scalar {variables}, while boldface lower-case letters $\mathbf{x}$ denote vectors.  $\{a, ..., b\}$ denotes the set of all integers between $a$ and $b$ and $\emptyset$ denotes an empty set. $|\cdot|$ denotes the absolute value operator. $\mathbb{R}^{N}$ represents the set of all $N$-dimensional vectors with real-valued entries. For a vector {$\mathbf{x}=(x_1,...,x_N)^T\in\mathbb{R}^{N}$}, $||\mathbf{x}||_2$ denotes the Euclidean norm, whereas  $\mathbf{x}^T$ stands for the  transpose of $\mathbf{x}$.  For real-valued multivariate functions $f(\mathbf{x})$ and $g(\mathbf{x})$, $\frac{\partial f}{\partial \mathbf{x}}(\mathbf{a})=\Big(\frac{\partial f}{\partial x_1}(\mathbf{a}),...,\frac{\partial f}{\partial x_N}(\mathbf{a})\Big)^T$ denotes the partial derivative of  $f$ \ac{wrt} $\mathbf{x}$ evaluated for an arbitrary vector $\mathbf{a}$ and $f(x) \circledast g(x)$ refers to  the convolution operation. For real numbers $a$ and $b$, $\max(a,b)$ and $\min(a,b)$ stand for the maximum and minimum of $a$ and $b$, respectively. For a scalar $x \in \mathbb{R}$, $[x]^+$ refers to $\max(0,x)$. The notation $X_n(\mathbf{x}_1,...,\mathbf{x}_i)$ highlights that $X$ depends on optimization variables $( \mathbf{x}_1,...,\mathbf{x}_i)$ and time slot $n$. {In Table \ref{tab:variables}, we present the most important variables used in the paper.} 
\begin{table}[]
	\centering
	\caption{List of variables.}
	\label{tab:variables}
	\setlength\tabcolsep{1.5pt} % default value: 6pt
	\begin{tabular}{|c|c|}
		\hline
		\rowcolor[HTML]{EFEFEF} 
		 Variable &  Definition \\ \hline
		 $\mathbf{p}_t[n]$ &  Coordinates of the reference ground target \\ \hline
		 $\mathbf{v}$ &  \ac{uav} velocity vector in \ac{3d} space \\ \hline
		 $v_y[n]$ &  \ac{uav} velocity along $y$-axis in time slot $n$ \\ \hline
		 $\mathbf{q}_i[n]$ &  \ac{3d} position of UAV $U_i, i\in \{1,2\},$ in time slot $n$ \\ \hline
		 $\mathbf{q}_i$ &  \ac{2d} position of UAV $U_i, i\in \{1,2\},$ in $xz$-plane \\ \hline
		 $b$ and $b_{\bot}$ &  Interferometric baseline and its perpendicular component \\ \hline
		 $\theta_2$ &  Radar look angle of  slave \ac{uav} $U_2$ \\ \hline
		 $\alpha$ &  Angle between the baseline and the horizontal plane \\ \hline
		 $r_i$ &  Radar slant range of UAV $U_i, i\in \{1,2\},$ \ac{wrt} $\mathbf{p}_t$ \\ \hline
		 $S$ &  Common radar swath width between $U_1$ and $U_2$ \\ \hline
		 $C_N$ &  Total \ac{insar} coverage \\ \hline
		 $\mathrm{SNR}_{i,n}$ &  \ac{snr} achieved by UAV $U_i, i\in \{1,2\},$ in time slot $n$ \\ \hline
		 $\gamma_{\mathrm{SNR},n}$ &  \ac{snr} decorrelation in time slot $n$ \\ \hline
		 $\rm \gamma_{\rm Rg}$ &  Baseline decorrelation \\ \hline
		 $\mathrm{\gamma}_n$ &  Total \ac{insar} coherence in time slot $n$\\ \hline
		 $h_{\rm amb}$ &  Height of ambiguity \\ \hline
		 $\Delta \phi_{90\%}$ &  90\% percentile of the random phase error \\ \hline
		 $\Delta h_{90\%}$ &  Point-to-point relative height error \\ \hline
		 $P_{\mathrm{com},i}[n]$ &  Instantaneous communication power consumed by $U_i$ \\ \hline
		 $d_{i,n}$ &  Distance $U_i, i\in \{1,2\},$ to \ac{gs} in time slot $n$ \\ \hline
		 $R_{i,n}$ &  Instantaneous throughput from $U_i, i\in \{1,2\},$ to \ac{gs} \\ \hline
		 $R_{\mathrm{min},i}$ &  \ac{sar} sensing data rate achieved by $U_i, i\in \{1,2\}.$ \\ \hline
		 $P_{\mathrm{prop},n}$ &  \ac{uav} propulsion power in time slot $n$ \\ \hline
		 $E_{i,N}$ &  Total energy consumed by $U_i, i\in \{1,2\},$ in time slot $n$ \\ \hline
	\end{tabular}
\end{table}
\section{System Model} \label{Sec:SystemModel}
\begin{figure}
	\centering
	\ifonecolumn
    \includegraphics[width=3.5in]{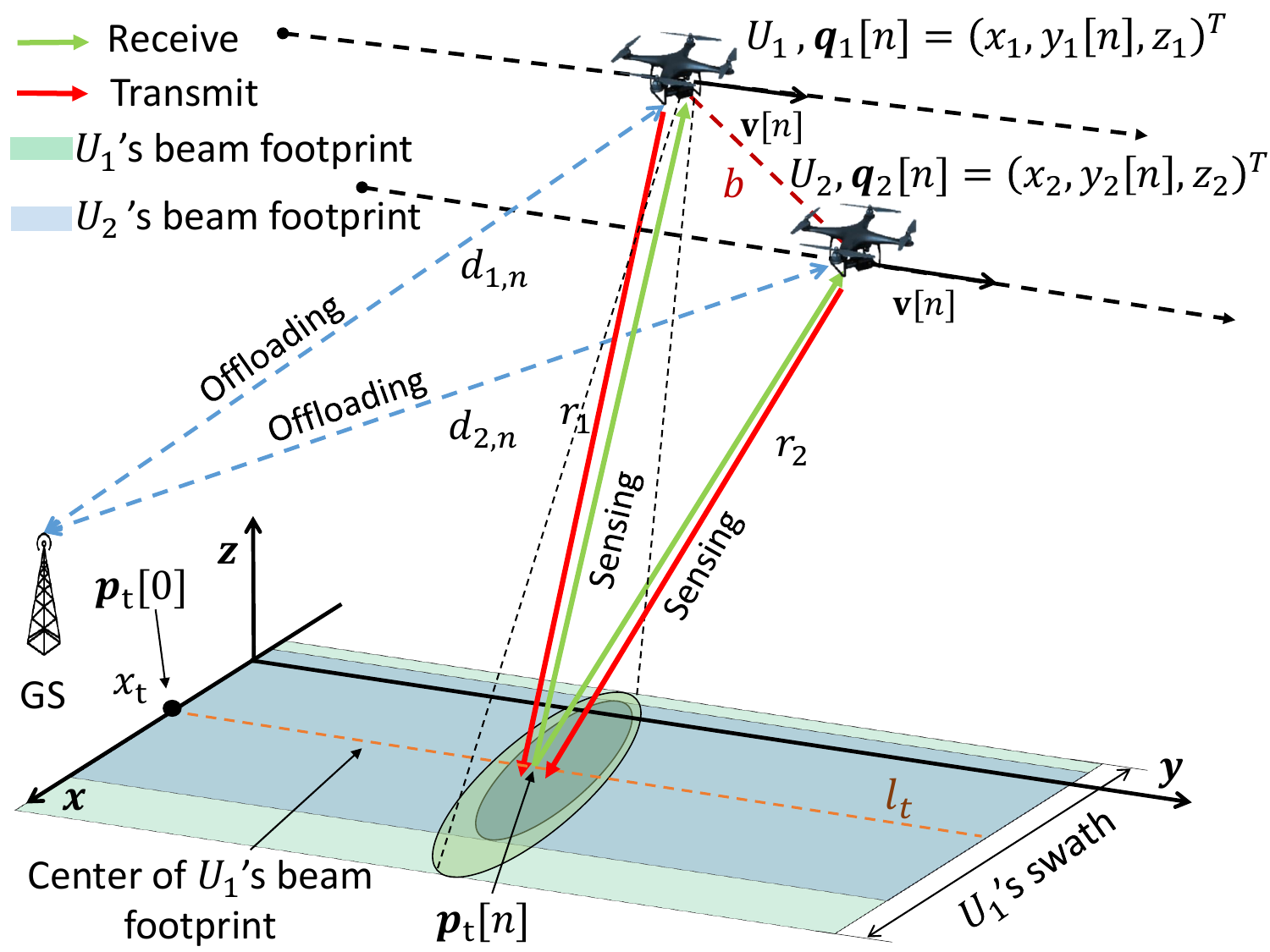}
    \else
    \includegraphics[width=0.9\columnwidth]{figures/SystemModel.pdf}
    \fi
    \caption{\ac{insar} sensing system with two \ac{uav} \ac{sar} sensors and a GS for real-time data offloading.}
    \label{fig:system-model}
\end{figure}
 {We consider two rotary-wing\footnote{The results presented in this paper can be  extended to fixed-wing \acp{uav} by considering a corresponding propulsion power model.} \acp{uav}, denoted by $U_1$ and $U_2$, that perform \ac{insar} sensing over a given area of interest. The two \ac{uav}-\ac{sar} systems move along separate linear trajectories above the target area, and each \ac{uav} captures a different \ac{sar} image of the same ground area. The two \ac{sar} images obtained are then combined using interferometry techniques, where the phase difference between both images is exploited to generate an accurate \ac{3d} \ac{dem}.}   We adopt a \ac{3d} coordinate system, where the $x$-axis represents the ground range direction, the $y$-axis represents the azimuth direction, and the $z$-axis defines the flying altitude of the drones above ground, see Figure  \ref{fig:system-model}. We discretize the total mission time $T$ into $N$ uniform time slots of duration $\delta_t$ such that  $T=N \cdot \delta_t$.
 Both drones follow a linear trajectory as the stripmap \ac{sar} imaging mode, typically used in \ac{insar} systems, is employed \cite{fixed_heading}. Therefore, the radar coverage along the $x$-axis, referred to as  swath, is centered \ac{wrt} a line {$l_{t}$} that is parallel to the $y$-axis and passes, {in time slot $n \in \{1, ..., N\}$}, through reference point $\mathbf{p}_t[n]=(x_t,y[n],0)^T \in \mathbb{R}^3$, see Figure \ref{fig:system-model}. Without loss of generality, $U_1$ is the farthest \ac{uav} from $ \mathbf{p}_t[n], \forall n,$ and is referred to as the master drone, whereas  $U_2$ is the slave drone.  Moreover, we {consider} across-track interferometry \cite{insar_introduction}, where both drones transmit and receive echoes and are located in the same  $xz$-plane, also referred to as the across-track plane \cite{insar_introduction}. They fly with the same velocity $\mathbf{v}_y = (v_y[1], ...,v_y[N])^T\in\mathbb{R}^{N}$, such that in time slot {$n$}, the instantaneous velocity vector is given by $\mathbf{v}[n]=(0,v_y[n],0)^T\in\mathbb{R}^3, \forall n$,  \cite{snr_equation}. The location of $U_i$ in time slot $n$,  $ i \in \{1,2\}$,  is denoted by $\mathbf{q}_i[n]=(x_i,y[n],z_i)^T$, where the $y$-axis position vector $\mathbf{y}=(y[1]=0,y[2], ..., y[N])^T\in\mathbb{R}^{N}$ is given by: 
 \begin{align}
 y[n+1]=y[n]+v_y[n]\delta_t, \forall n \in \{1,N-1\}.
 \end{align}
Hereinafter, for ease of notation, we use $\mathbf{q}_i=(x_i,z_i)^T \in \mathbb{R}^2, \forall i \in \{ 1,2\}$, to denote the position of $U_i$ in the across-track plane (i.e., $xz-$plane), {see Figure \ref{fig:system-model-2},} whereas $\mathbf{q}_i[n]$ refers to the full \ac{3d} position of \ac{uav}  $U_i$ in time slot $n$. Furthermore, the interferometric baseline, which refers to the distance between the two \ac{insar} sensors, is given by:
\begin{align}
   b(\mathbf{q}_1,\mathbf{q}_2) = ||\mathbf{q}_2 -\mathbf{q}_1||_2.
\end{align}
\begin{figure}
	\centering
	\ifonecolumn
	\includegraphics[width=4in]{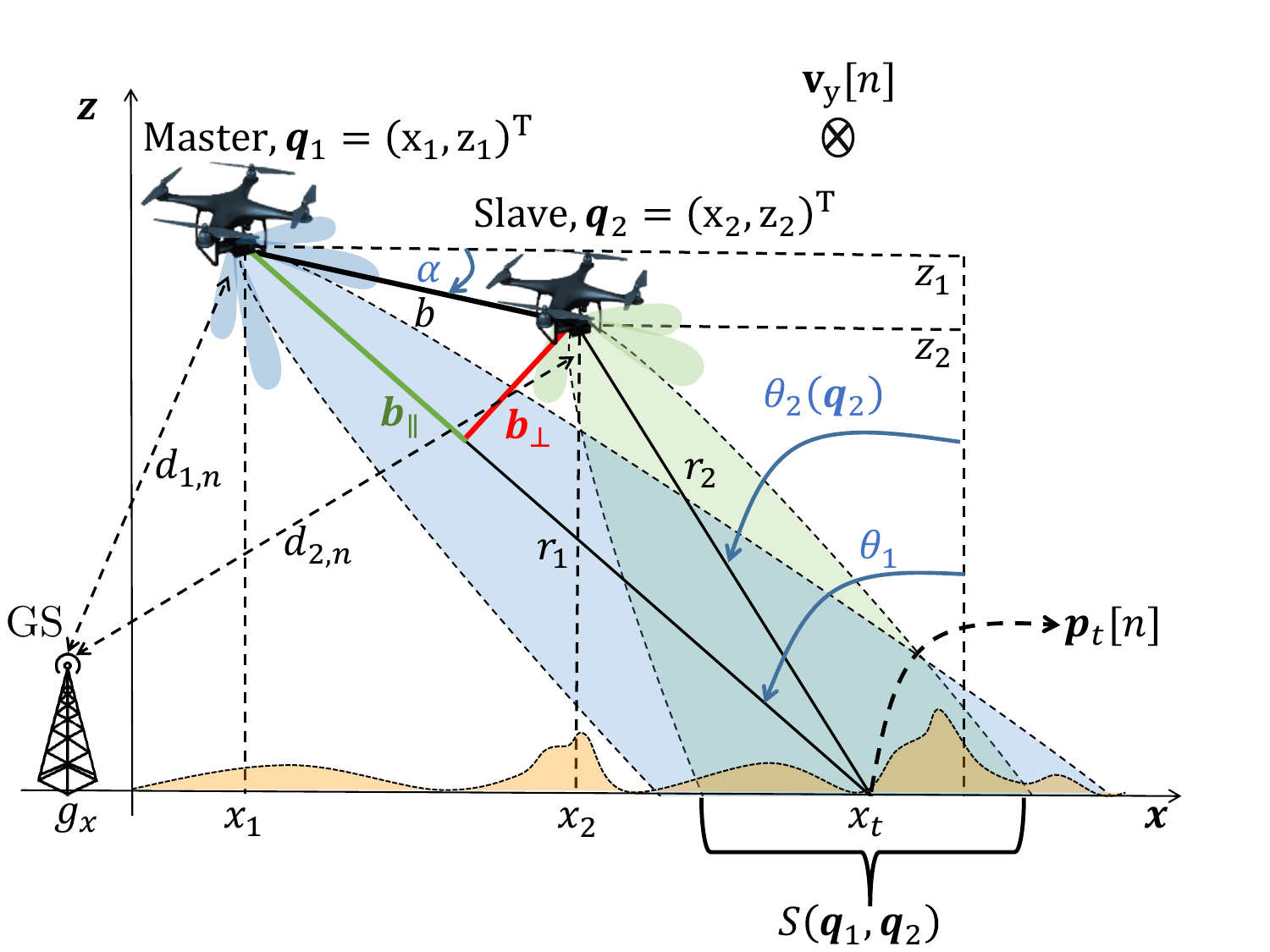}
	\else 
	\includegraphics[width=0.9\columnwidth]{figures/SystemModel2.pdf}
	\fi
	\caption{Illustration of the bistatic \ac{uav} formation, denoted by $\{\mathbf{q}_1,\mathbf{q}_2\}$, in the across-track plane (i.e., $xz-$plane).}
	\label{fig:system-model-2}
\end{figure}
 The perpendicular baseline, denoted by $b_{\bot}$, is the magnitude of the projection of the baseline vector perpendicular to $U_1$'s \ac{los} { to $\mathbf{p}_t[n]$} and is given by:
\begin{equation}
	b_{\bot}(\mathbf{q}_1,\mathbf{q}_2)=
	b(\mathbf{q}_1,\mathbf{q}_2)  \cos\Big(\theta_1- \alpha(\mathbf{q}_1,\mathbf{q}_2)\Big),
\end{equation} 
where $\theta_1$ is the fixed angle that $U_1$'s \ac{los} to $\mathbf{p}_t[n]$ has with the vertical, i.e., the look angle of the master \ac{uav}, and $\alpha(\mathbf{q}_1,\mathbf{q}_2)$ is the angle between the interferometric baseline and the horizontal plane, see Figure \ref{fig:system-model-2}. 
\subsection{Bistatic \ac{insar} Coverage}
Let $r_1$ and $r_2$ denote the slant range of the radars of $U_1$ and $U_2$ \ac{wrt} $ \mathbf{p}_t[n]$,  respectively. These slant ranges are {independent of time and} given by:
\begin{align}
	r_i(\mathbf{q}_i)= \sqrt{ (x_i -x_t)^2 +  z_i^2  }, \forall i \in \{1,2\}.
\end{align} To maximize the coverage, the master  \ac{uav} is positioned in the across-track plane such that its \ac{sar} antenna beam footprint on the ground is centered around {$\mathbf{p}_t[n]$}. Furthermore, without loss of generality, we assume  $r_2(\mathbf{q}_2)\leq    r_1(\mathbf{q}_1)$ and impose the following condition: 
\begin{align}
	x_1=x_t - z_1 \tan(\theta_1). \label{eq:master_drone_placement}
\end{align} 
As the look angle of the master platform is fixed to $\theta_1$, we propose to adjust the look angle of the slave \ac{uav}, denoted by $\theta_2(\mathbf{q}_2)$, such that its beam footprint is centered around  {$\mathbf{p}_t[n]$}, i.e., $\theta_2(\mathbf{q}_2)= \arctan\left(\frac{x_2-x_t}{z_2}\right)$. The \ac{insar} coverage is limited to the area where the beam footprints of $U_1$ and $U_2$ overlap, see Figure \ref{fig:system-model-2}. In this work, we focus on a single scan (i.e., single swath) for the two drones, where the usable swath can be obtained as follows:
\ifonecolumn 
\begin{align}\label{eq:swath_width}
  S(\mathbf{q}_1,\mathbf{q}_2 ) = \Big[ \min\left(S_{\rm far}(\mathbf{q}_1),S_{\rm far}(\mathbf{q}_2)\right)-    
     \max\left(S_{\rm near}(\mathbf{q}_1) ,S_{\rm near}(\mathbf{q}_2) \right)\Big]^+,
\end{align}
\else 
\begin{equation}\label{eq:swath_width}
	S(\mathbf{q}_1,\mathbf{q}_2 ) = \Big[ \min_{i \in \{1,2\}}\left\{S_{\rm far}(\mathbf{q}_i)\right\}-  \max_{i \in \{1,2\}}\left\{S_{\rm near}(\mathbf{q}_i)\right\}\Big]^+,
\end{equation}
\fi
where $S_{\rm far}(\mathbf{q}_i)=x_i+z_i \tan(\theta_i(\mathbf{q}_i) + \frac{\Theta_{\mathrm{3dB}}}{2})$, $S_{\rm near}(\mathbf{q}_i)=x_i+z_i \tan(\theta_i(\mathbf{q}_i) - \frac{\Theta_{\mathrm{3dB}}}{2})$, $\forall i \in \{1,2\}$, $\theta_1(\mathbf{q}_1)=\theta_1$, and  $\Theta_{\mathrm{3dB}}$ is the -3 dB beamwidth in elevation. The total area covered by the \ac{insar} radar is approximated\footnote{The approximation is due to the elliptical shape of the beam footprint on the ground and becomes negligible for large $N$. } as follows:
\begin{align} \label{eq:insar_coverage}
C_N(\mathbf{q}_1,\mathbf{q}_2,\mathbf{v}_y) =\sum^{N-1}_{n=1} S(\mathbf{q}_1,\mathbf{q}_2 ) v_y[n] \delta_t.
\end{align}
\subsection{\ac{insar} Performance Metrics}
{In practical \ac{insar} missions, such as TanDEM-X and SRTM \cite{coherence1}, the performance requirements for \acp{dem} are typically defined in terms of coverage and height accuracy. The height accuracy, i.e., the precision with which the ground topography can be estimated, depends on several intermediate performance metrics, specifically the coherence and the \ac{hoa}. In the following, we introduce these key \ac{3d} sensing performance metrics and highlight their impact on the overall \ac{insar} performance.}
\subsubsection{Total \Ac{insar} Coherence} 
{The interferometric coherence, denoted by $\gamma_n\in[0,1]$, represents the cross-correlation between the master and slave \ac{sar} images. In general, coherence values close to 1 result in improved sensing accuracy \cite{coherence1}. The total coherence can be decomposed into several decorrelation sources as follows:
\begin{equation}
	\gamma_n= \gamma_{\rm Rg} \gamma_{\mathrm{SNR},n} \gamma_{\rm other}, \forall n,
\end{equation}
where $\gamma_{\rm Rg}\in[0,1]$ denotes the baseline decorrelation, $\gamma_{\mathrm{SNR},n}\in[0,1]$ is the \ac{snr} decorrelation, and $\gamma_{\rm other}\in[0,1]$ represents the contribution from all other decorrelation sources, e.g., volumetric decorrelation. Here, we assume that surface scattering is dominant and $\gamma_{\rm other}$ is independent of the formation.  \cite{coherence1}, and that no prior knowledge about the imaged scene is available.}
\subsubsection{\ac{snr} Decorrelation} 
Low \acp{snr}  in \ac{sar} data acquisition result in a loss of the coherence between the master and slave \ac{sar} images during processing, {causing a degradation of the sensing accuracy.} In time slot $n$, the corresponding \ac{snr} decorrelation is given by  \cite{snr_equation}: 
 \begin{align} \label{eq:snr_decorrelation}
 &\gamma_{\mathrm{SNR},n}(\mathbf{q}_1,\mathbf{q}_2,\mathbf{v}_y)= \prod_{i \in \{1,2\}} \frac{1}{\sqrt{1+\mathrm{SNR}^{-1}_{i,n}(\mathbf{q}_i,\mathbf{v}_y)}}, \forall n, 
 \end{align}
 where $\mathrm{SNR}_{i,n}$ denotes the \ac{snr} achieved by $U_i$ in time slot $n$ and is given by \cite{snr_equation}:
\begin{equation}
\mathrm{SNR}_{i,n}(\mathbf{q}_i,\mathbf{v}_y)=\frac{\gamma_{r,i}}{v_y[n] r_i^3(\mathbf{q}_i)  \sin(\theta_i(\mathbf{q}_i)) }  , \forall n,
  \end{equation}
where $\gamma_{r,i}=\frac{\sigma_0 P_{t,i}\; G_t\; G_r \; \lambda^3\; c \;\tau_p \;\mathrm{PRF}}{4^4 \pi^3   k_b T_{\mathrm{sys}} \; B_{\mathrm{Rg}} \; F \; L_{\mathrm{atm}} \; L_{\mathrm{sys}} \; L_{\mathrm{az}}}, \forall i \in \{1,2\}$. Here, $\sigma_0$ is the normalized backscatter coefficient, $P_{t,i}$ is $U_i$'s radar transmit power, $G_t$ and $G_r$ {are} the transmit and receive  antenna gains, respectively, $\lambda$ is the wavelength, $c$ is the speed of light, $\tau_p$ is the pulse duration, $\mathrm{ PRF}$ is the pulse repetition frequency, $k_b$ is the Boltzmann constant, $T_{\mathrm{sys}}$ is the receiver  {noise} temperature, $B_{\mathrm{Rg}}$ is the bandwidth of the radar pulse, $F$ is the noise figure, and  $L_{\mathrm{atm}}$, $L_{\mathrm{ sys}}$, and $L_{\mathrm{ az}}$ represent the atmospheric, system, and azimuth losses, respectively. {Note that, in (\ref{eq:snr_decorrelation}), low master and slave \acp{snr} lead to small $\gamma_{\mathrm{SNR},n}$, and thus, degraded coherence $\gamma_n$.}
\begin{comment}
To guarantee a decent interferometric performance, we impose a minimum threshold denoted by $\gamma_{\mathrm{SNR}}^{\mathrm{min}}$ on the limited \ac{snr} decorrelation  such that: 
\begin{align}
    \gamma_{\mathrm{SNR}}[n](\mathbf{q}_1,\mathbf{q}_2) \geq  \gamma_{\mathrm{SNR}}^{\mathrm{min}}, \forall n,
\end{align}
\end{comment}
\subsubsection{Baseline Decorrelation} 
{The baseline decorrelation $\gamma_{\rm Rg}$ reflects the loss of coherence caused by the different incidence angles used for acquisition of the two \ac{sar} images. The baseline decorrelation is  given by \cite{victor}: }
\begin{equation} \gamma_{\mathrm{Rg}}(\mathbf{q}_2)=\frac{(2+B_{p}) \sin\left( \theta_2(\mathbf{q}_2) \right) - (2-B_{ p}) \sin(\theta_1)}{B_p \big(\sin(\theta_1)+\sin\left( \theta_2(\mathbf{q}_2) \right)\big)},\label{eq:baseline_decorrelation}
\end{equation}
where $B_{ p}=\frac{B_{\mathrm{Rg}}}{f_0}$ is the fractional bandwidth and $f_0$ is the radar center frequency. It can be shown  {that \ac{uav} formations with large inter-\ac{uav} distance, which corresponds to large perpendicular baselines $b_{\bot}$, cause high baseline decorrelation (i.e., $\gamma_{\rm Rg} \to 0$) and, therefore, degrade the overall coherence.}
\subsubsection{Height of Ambiguity (HoA)} The \ac{hoa} is defined as the height difference that results in a complete $2\pi$ cycle of the interferometric phase \cite{snr_equation}. It is therefore related to the sensitivity of the radar system to topographic height variations. Similar to the baseline decorrelation,  the \ac{hoa} depends on the \ac{uav} formation and is given by \cite{snr_equation}: 
\begin{align}\label{eq:height_of_ambiguity}
    h_{\mathrm{amb}}(\mathbf{q}_1,\mathbf{q}_2)=\frac{\lambda r_1(\mathbf{q}_1) \sin(\theta_1)}{b_{\perp}(\mathbf{q}_1,\mathbf{q}_2)}.
\end{align}
{Note that large perpendicular baselines $b_{\bot}$ lead to small \ac{hoa} values, which result in increased errors for interferometric phase unwrapping \cite{phase_unrwapping}. However, very small \ac{uav} baselines, i.e., large \acp{hoa}, are not a good choice either, since they lead to low sensing accuracy, as we explain in the following.}
\subsubsection{Sensing Height Accuracy}
A common statistics-based performance metric for assessing the accuracy of the height estimation in \ac{dem} is the point-to-point 90\% relative height  error given by \cite{snr_equation}: 
\begin{equation}
	\Delta h_{90\%}(\mathbf{q}_1,\mathbf{q}_2,\gamma_n) =h_{\mathrm{amb}}(\mathbf{q}_1,\mathbf{q}_2) \frac{\Delta\phi_{90\%}(\gamma_n) }{2\pi},
\end{equation}
where $\Delta\phi_{90\%}$ is defined as the  90\% percentile of the random phase error $\Delta \phi$ and obtained by evaluating the difference between two random variables, each of them describing the fluctuation of the interferometric phase within one interferogram as follows \cite{snr_equation}: 
\begin{equation} \label{eq:integral}
	\int_{-\Delta\phi_{90\%}(\gamma_n)}^{{\Delta\phi_{90\%}(\gamma_n)}} 	p_{\Delta \phi }(\Delta \phi,\gamma_n)  \circledast	p_{\Delta \phi }(\Delta \phi,\gamma_n)\cdot d\Delta \phi=0.9,
\end{equation}   {where $p_{\Delta \phi }$ is the  \ac{pdf}  of the phase error $\Delta \phi$ for a given interferometric coherence $\gamma$ and is given by}\footnote{In \cite{multilooking}, the \ac{pdf} is derived by assuming Gaussian scattering-matrix statistics. Yet, the proposed solution can be extended to other models, e.g., the K-distribution model \cite{K-distribution}. }  \cite{multilooking}:
\ifonecolumn
\begin{equation}
		\label{eq:pdf}
	p_{\Delta \phi }(\Delta \phi,\gamma_n)= \frac{\Gamma(n_L+0.5) (1 - |\gamma_n|^2)^{n_L}  |\gamma_n| \cos(\Delta \phi)}{ 2\sqrt{\pi} \Gamma(n_L) \left(1-|\gamma_n|^2\cos^2(\Delta \phi)\right)^{n_L+0.5} }+ \frac{(1-|\gamma_n|^2)^{n_L}}{2 \pi} F_{G}\left(n_L,1;0.5;|\gamma_n|^2 \cos^2(\Delta \phi)\right).
\end{equation}
\else 
\begin{align}
	\label{eq:pdf}
	p_{\Delta \phi }(\Delta \phi,\gamma_n)= \frac{\Gamma(n_L+0.5) (1 - |\gamma_n|^2)^{n_L}  |\gamma_n| \cos(\Delta \phi)}{ 2\sqrt{\pi} \Gamma(n_L) \left(1-|\gamma_n|^2\cos^2(\Delta \phi)\right)^{n_L+0.5} }+\notag \\ \frac{(1-|\gamma_n|^2)^{n_L}}{2 \pi} F_{G}\left(n_L,1;0.5;|\gamma_n|^2 \cos^2(\Delta \phi)\right).
\end{align}
\fi
Here, $\Gamma$ is the Gamma function, {$F_{G}$} is the Gauss hypergeometric function \cite{snr_equation}, and $n_L$ is the number of independent looks. In fact, multi-looking is a technique that averages adjacent pixels of the interferogram, {e.g.,} by filtering with a boxcar window, to improve phase estimation leading to a  trade-off between spatial resolution and height accuracy, see \cite{multilooking}. Note that  large \acp{hoa} result in large  point-to-point 90\% height errors $\Delta h_{90\%}$, and thus, in low relative height estimation accuracy. This explains why  optimizing the \ac{uav} baseline, {and} therefore, the \ac{hoa}, is  an intricate {problem}.
\subsection{Communication Performance}
We consider real-time offloading of the radar data to a \ac{gs}, where the master and slave \acp{uav} employ \ac{fdma}. {The  instantaneous   communication transmit power  consumed by \ac{uav} $U_i$ is given by $\mathbf{P}_{\mathrm{com},i}=(P_{\mathrm{com},i}[1],...,P_{\mathrm{com},i}[N])^T \in \mathbb{R}^N, i \in \{1,2\}$.}
We denote the location of the \ac{gs} by $\mathbf{g}= (g_x, g_y, g_z)^T \in \mathbb{R}^3$ and the distance from $U_i$ to the \ac{gs} by    $d_{i,n}(\mathbf{q}_i,\mathbf{v}_y) = ||\mathbf{q}_i[n]-\mathbf{g} ||_2, \forall i \in \{1,2\}, \forall n.$
We suppose that both \acp{uav} fly at   sufficiently high  altitudes to allow
obstacle-free communication with the \ac{gs} over a \ac{los} link. Thus, based on the free-space path loss model and \ac{fdma}, the
instantaneous throughput  from $U_i, \forall i\in \{1,2\},$ to the \ac{gs} is given by:
\begin{align}
 &R_{i,n}(\mathbf{q}_i,\mathbf{P}_{\mathrm{com},i},\mathbf{v}_y)= B_{c,i} \; \log_2\left(1+\frac{P_{\mathrm{com},i}[n] \;\beta_{c,i}}{d_{i,n}^2(\mathbf{q}_i,\mathbf{v}_y)}\right), \forall n,
\end{align}
where  $B_{c,i}$ is the fixed communication bandwidth allocated for $U_i$ and $\beta_{c,i}$ is its reference channel gain\footnote{The reference channel gain is the channel power gain at a reference distance of 1 m.} divided by the noise variance. The minimum {required} data rate, usually fixed in communication systems, is significantly impacted by the {formation} of the two UAVs and depends on the imageable area size. Based on \cite{data_rate}, the \ac{sar} sensing data rate  for $U_i, \forall i \in \{1,2\},$ is given by: 
\ifonecolumn
\begin{equation}
	R_{\mathrm{min},i}(\mathbf{q}_i)=  n_B  B_{\mathrm{Rg}}\mathrm{PRF}\left(\frac{z_i}{c} \left( \frac{1}{\cos \left(\theta_i(\mathbf{q}_i) + \frac{\Theta_{\mathrm{3dB}}}{2} \right)}  -\frac{1}{\cos \left(\theta_i(\mathbf{q}_i) - \frac{\Theta_{\mathrm{3dB}}}{2} \right)} \right) + \tau_p\right),
\end{equation}
\else
\begin{align}
	&R_{\mathrm{min},i}(\mathbf{q}_i)=  n_B  B_{\mathrm{Rg}}\mathrm{PRF}\Bigg(  \tau_p+\notag \\ & \frac{z_i}{c} \Bigg(\frac{1}{\cos \left(\theta_i(\mathbf{q}_i) + \frac{\Theta_{\mathrm{3dB}}}{2} \right)}  -\frac{1}{\cos \left(\theta_i(\mathbf{q}_i) - \frac{\Theta_{\mathrm{3dB}}}{2} \right)} \Bigg)\Bigg),
\end{align}
\fi
where $n_B$ is  the number of bits per complex sample. 
\subsection{Energy Consumption}
The propulsion power $P_{\mathrm{prop},n}$  consumed by the drone in time slot $n$ is given by \cite{propulsion}: 
\ifonecolumn
\begin{equation}\label{eq:propulsion_power}
	P_{\mathrm{prop},n}(\mathbf{v}_y)=  P_0 \left(1+\frac{3v_y^2[n]}{U^2_{\mathrm{tip}}}\right)+P_I\left( \sqrt{1+\frac{v^4_y[n]}{4v_0^4}}-\frac{v^2_y[n]}{2v_0^2}\right)^{\frac{1}{2}}  + \frac{1}{2}d_0 \rho s A_e v^3_y[n], \forall n.
\end{equation}
\else
\begin{align}\label{eq:propulsion_power}
	P_{\mathrm{prop},n}(\mathbf{v}_y)&=  P_0 \left(1+\frac{3v_y^2[n]}{U^2_{\mathrm{tip}}}\right)+ \frac{1}{2}d_0 \rho s A_e v^3_y[n]+ \notag \\ &P_I\left( \sqrt{1+\frac{v^4_y[n]}{4v_0^4}}-\frac{v^2_y[n]}{2v_0^2}\right)^{\frac{1}{2}}, \forall n.
\end{align}
\fi
Here, $P_0=\frac{\delta_u}{8}\rho s A_e \Omega^3 R^3$ and  $P_I=(1+k_u) \frac{W_u^{\frac{3}{2}}}{\sqrt{2\rho A_e}}$ are two constants, $v_0=\sqrt{\frac{W_u}{2 \rho A_e}}$ is the mean rotor induced velocity {while hovering}, $U_{\mathrm{tip}}$ is the tip speed of the rotor blade, $d_0$ is the fuselage drag ratio,  $\delta_u$ is the profile drag coefficient, $\rho$ is the air density, $s$ is the rotor solidity, $A_e$ is the rotor disc area, $\Omega$ is the blade angular velocity, $R$ is the rotor radius, $k_u$ is a  correction factor, and $W_u$ is the aircraft weight in Newton. The total energy  consumed by \ac{uav} $U_i, \forall i \in \{1,2\},$ is given by: 
\begin{equation}
	E_{i,N}(\mathbf{P}_{\mathrm{com},i}, \mathbf{v}_y)=\sum\limits_{n=1}^{N} \delta_t  \Big (P_{\mathrm{prop},n}( \mathbf{v}_y) +  P_{t,i}+ P_{\mathrm{com},i}[n] \Big).
\end{equation}

\section{Problem Formulation}
In this paper, we aim to maximize the  \ac{insar} coverage by jointly optimizing the \ac{uav} formation $\{\mathbf{q}_1,\mathbf{q}_2\}$, the instantaneous communication transmit powers $\{\mathbf{P}_{\mathrm{com,1}},\mathbf{P}_{\mathrm{com,2}} \}$, and the instantaneous \ac{uav} velocity $\mathbf{v}_y$,  while satisfying energy, communication, and interferometric quality-of-service constraints. To this end, we formulate  the following optimization problem: 
\begin{alignat*}{2} 
&(\mathrm{P.1}):\max_{\mathbf{q}_1,\mathbf{q}_2, \mathbf{P}_{\mathrm{com},1}, \mathbf{P}_{\mathrm{com},2},\mathbf{v}_y} \hspace{3mm}  C_N(\mathbf{q}_1,\mathbf{q}_2,\mathbf{v}_y)   & \qquad&  \\
\text{s.t.} \hspace{3mm} &  \mathrm{C1: } \; z_{\mathrm{min}} \leq z_i \leq z_{\mathrm{ max}}, \forall i \in \{ 1,2\},               &      &  \\ & \mathrm{C2}: \;  x_1=x_t - z_1 \tan(\theta_1),             &      &  \\& \mathrm{C3}: \;  r_2(\mathbf{q}_2) \leq r_1(\mathbf{q}_1),             &      &  \\& \mathrm{C4}: \;    x_2 \leq x_t,             &      &  \\&  \mathrm{C5}:    b(\mathbf{q}_1,\mathbf{q}_2)  \geq b_{\mathrm{min}},  &      &     
 \\
  &  \mathrm{C6}:  \gamma_{\mathrm{SNR},n}(\mathbf{q}_1,\mathbf{q}_2,\mathbf{v}_y)\geq \gamma_{\mathrm{SNR}}^{\mathrm{min}}, \forall n,        &      &     
 \\
 &  \mathrm{C7}: \gamma_{\mathrm{Rg}}(\mathbf{q}_2) \geq \gamma_{\mathrm{Rg}}^{\mathrm{min}},            &      &     
 \\
 &    \mathrm{C8}: \;  h_{\mathrm{amb}}(\mathbf{q}_1,\mathbf{q}_2)\geq  h_{\mathrm{amb}}^{\mathrm{min}}  ,           &      & \\
 &    \mathrm{C9}: \;    \Delta h_{90\%}^{\rm worst}(\mathbf{q}_1,\mathbf{q}_2) \leq \Delta h^{\mathrm{max}},           &      & \\
& \mathrm{C10}: 0 \leq P_{\mathrm{com},i}[n]  \leq P_{\mathrm{com}}^{\mathrm{max}}, \forall \; i \in \{1,2\}, \forall n,    & &\\
& \mathrm{C11}: R_{i,n}(\mathbf{q}_i,\mathbf{P}_{\mathrm{com},i},\mathbf{v}_y) \geq R_{\mathrm{min},i} (\mathbf{q}_i), \forall \; i \in \{1,2\}, \forall n,       & &  \\
& \mathrm{C12}: 	E_{i,N}(\mathbf{P}_{\mathrm{com},i}, \mathbf{v}_y) \leq E_{\mathrm{max},i}, \forall \; i \in \{1,2\},    & &  \\
& \mathrm{C13}: v_{\mathrm{min}} \leq v_y[n] \leq v_{\mathrm{max}}, \forall n,    & & \\
& \mathrm{C14}: \theta_{\mathrm{min}} \leq \theta_2(\mathbf{q}_2) \leq \theta_{\mathrm{max}},   & & \\
& \mathrm{C15}: \theta_2(\mathbf{q}_2)= \arctan\left(\frac{x_2-x_t}{z_2}\right).  & & 
\end{alignat*}
Constraint $\mathrm{ C1}$ specifies the maximum and minimum allowed flying altitude, denoted by $z_{\mathrm{max}}$ and $z_{\mathrm{min}}$,  respectively. Constraints $\mathrm{ C2}$ and $\mathrm{ C3}$ ensure maximum overlap between the beam footprint of the master drone and the area of interest. Constraint $\mathrm{ C4}$ is imposed because a side-looking \ac{sar} is assumed. Constraint $\mathrm{ C5}$ ensures a minimum separation distance {between the \acp{uav}}, denoted by $b_{\mathrm{min}}$. Constraints $\mathrm{ C6}$ and $\mathrm{ C7}$ ensure minimum required thresholds for the sensing $\mathrm{ SNR}$ and baseline decorrelation, denoted by $\gamma_{\mathrm{SNR}}^{\mathrm{min}}$ and $\gamma_{\mathrm{Rg}}^{\mathrm{min}}$,  respectively. Constraint $\mathrm{ C8}$ imposes a minimum \ac{hoa}, denoted by $h_{\mathrm{amb}}^{\mathrm{min}}$, that satisfies phase unwrapping requirements \cite{coherence1}.   Constraint $\mathrm{C9}$ ensures that the {worst-case} point-to-point 90\% relative height error of \ac{uav} formation $\{\mathbf{q}_1,\mathbf{q}_2\}$, defined as $\Delta h_{90\%}^{\rm worst}(\mathbf{q}_1,\mathbf{q}_2)=\Delta h_{90\%} (\mathbf{q}_1,\mathbf{q}_2,\gamma_{\rm SNR}^{\rm min}\gamma_{\rm Rg}^{\rm min}\gamma_{\rm other})$, does not exceed the maximum allowed height error, denoted by $\Delta h^{\rm max}$. The definition of $\Delta h_{90\%}^{\rm worst}$ is based on the fact that (i) an analytical expression for (\ref{eq:integral}) cannot be derived  and (ii) the  phase error variance decreases with the coherence \cite{multilooking}. Therefore, (\ref{eq:integral}) is solved numerically for the worst case scenario, as $\gamma_n\geq\gamma_{\rm SNR}^{\rm min}\gamma_{\rm Rg}^{\rm min}\gamma_{\rm other}, \forall n$.  Constraint $\mathrm{ C10}$ ensures that the communication transmit power is non-negative and does not exceed the maximum allowed level denoted by $P_{\mathrm{com}}^{\mathrm{max}}$.  Constraint $\mathrm{ C11}$ ensures that the achievable throughput of drone $U_i$ does not fall below the minimum {required} data rate $R_{\mathrm{min},i}(\mathbf{q}_i), \forall i \in \{1,2\}$, which depends on the imageable area, thus the \ac{uav} formation. Constraint $\mathrm{ C12}$ limits the  total energy consumed by \ac{uav} $U_i$ to its maximum battery capacity, denoted by $E_{\mathrm{max},i}, \forall i \in \{1,2\}$. Constraint  $\mathrm{C13}$ {specifies} the minimum and maximum allowed velocities of the \ac{uav}-\ac{sar} system, denoted by $v_{\rm min}$ and $v_{\rm max}$, respectively. { Constraint  $\mathrm{C14}$  {limits} the maximum and minimum slave look {angles to} $\theta_{\rm max}$ and $\theta_{\rm min}$, respectively. Constraint  $\mathrm{C15}$ {directs} the slave look angle towards $  \mathbf{p}_t[n], \forall n$. } { We note} that some {of the} constraints do not depend on time $n$ as some variables, such as $x_1$ and $z_2$, are optimized but are fixed across all time slots due to the prescribed linear \ac{insar} trajectory imposed by the stripmap mode \ac{sar}  operation \cite{fixed_heading}. \par 
Problem $\mathrm{(P.1)}$ is a non-smooth non-convex optimization problem for the following reasons. First, the use  of a variable slave look angle as well as $\max$ and $\min$ functions causes the objective function $C_N$ to be non-smooth and non-convex \ac{wrt} $\mathbf{q}_2$. Second, {the} coupling between different optimization variables and the presence of trigonometric functions in {the} expressions for the $\mathrm{SNR}$ decorrelation, baseline decorrelation, and \ac{hoa} (see (\ref{eq:snr_decorrelation}), (\ref{eq:baseline_decorrelation}), and (\ref{eq:height_of_ambiguity})), make  constraints $\mathrm{C6}$, $\mathrm{C7}$, $\mathrm{C8}$, and $\mathrm{C9}$  non-convex. Third, the use of variable minimum sensing data rates {causes} $\mathrm{C11}$ to be neither monotonic nor convex. Fourth, constraints $\mathrm{C3}$ and $\mathrm{C5}$ involve lower bounds on Euclidean distances{, and thus, they are} also non-convex. Lastly, the \ac{uav} propulsion power, given in (\ref{eq:propulsion_power}),  {causes} constraint $\mathrm{C12}$ {to be} non-convex. {In short, problem $\mathrm{(P.1)}$ involves interesting performance trade-offs, which do not allow for a trivial solution. For instance, achieving extensive coverage requires a high speed, which reduces the sensing \acp{snr}, leading to degraded coherence and, consequently, increased estimation errors. Furthermore, choosing a close flying formation for the \acp{uav} can improve the coverage (due to the increased overlap of the corresponding individual images), however, it also causes the baseline to be small, which leads to high \ac{hoa}, and consequently, large estimation errors.}  {In summary}, it is very difficult and challenging to find {a} globally optimal solution to problem $\mathrm{(P.1)}$. To overcome {this} challenge, in the next section, we divide problem $\mathrm{(P.1)}$ into several more tractable and easier-to-solve sub-problems and {tackle} them in an alternating manner.
\section{Solution of the Optimization Problem} 
\begin{figure}[]
	\centering
	\ifonecolumn
	\includegraphics[width=5in]{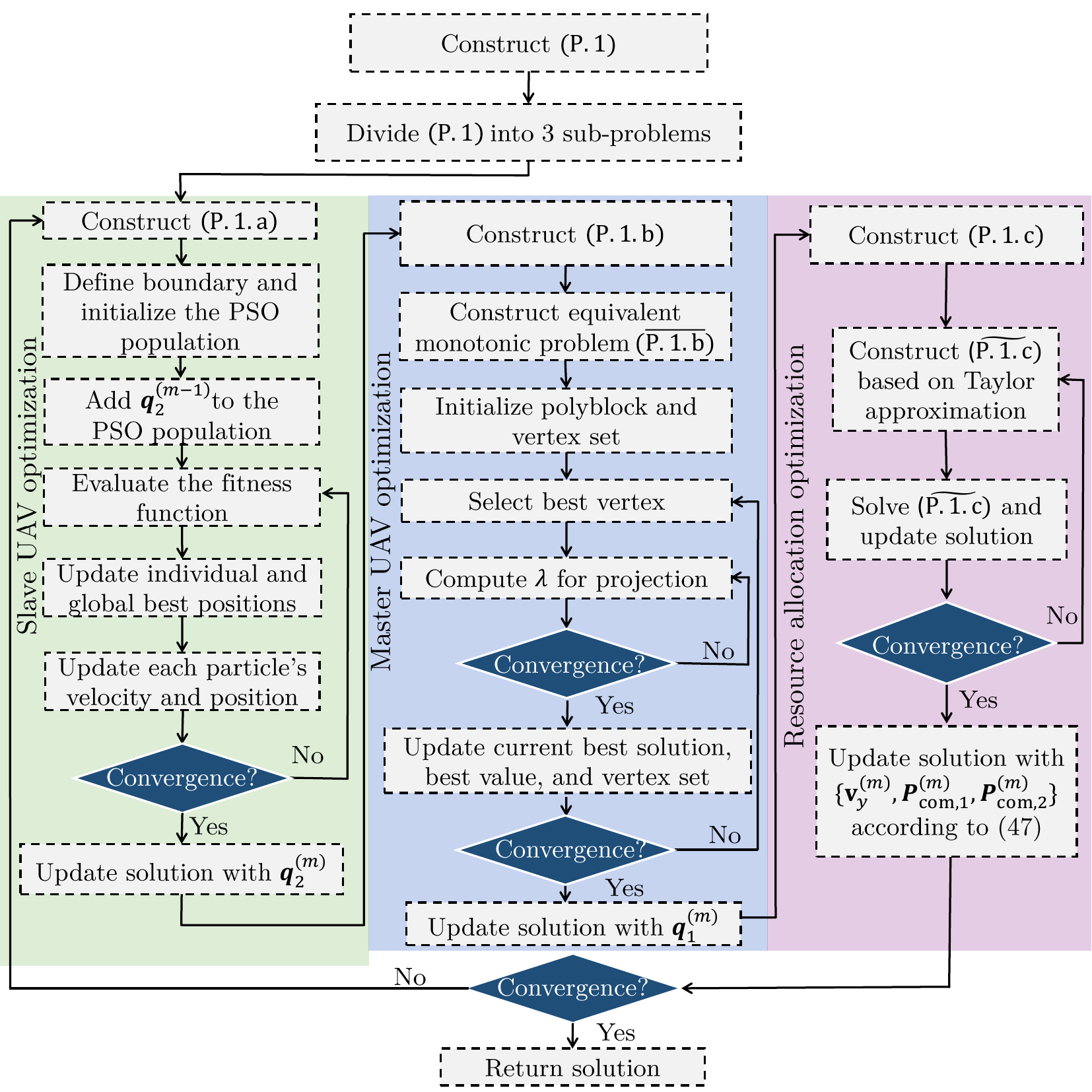}
	\else
	\includegraphics[width=1\columnwidth]{figures/Block_diagram.pdf}
	\fi
	\caption{Block diagram of the proposed solution to problem $\rm (P.1)$ based on \ac{ao}, \ac{sca}, polyblock outer approximation, and \ac{pso}.}
	\label{fig:block_diagram}
\end{figure}
To strike  a balance between performance and complexity,  we provide a low-complexity sub-optimal solution for the formulated problem based on \ac{ao}. To this end, { in each iteration $m$ of the proposed \ac{ao} algorithm,} problem $\mathrm{ (P.1)}$ is divided into three sub-problems, namely $\mathrm{(P.1.a)}$, $\mathrm{(P.1.b)}$, and $\mathrm{(P.1.c)}$, that are solved alternately{, as illustrated in Figure \ref{fig:block_diagram}}.  In problems $\mathrm{(P.1.a)}$ and  $\mathrm{(P.1.b)}$, we optimize the {positions $\mathbf{q}_2$  and $\mathbf{q}_1$ of the} slave and master \acp{uav} in the across-track plane, respectively. In sub-problem $\mathrm{(P.1.c)}$, we optimize the \ac{uav} resources, {namely} the instantaneous velocity $\mathbf{v}_y$  {(which affects} the propulsion power $\mathbf{P}_{\rm prop}${)} and the instantaneous communication transmit powers, $\mathbf{P}_{\rm com,1}$ and  $\mathbf{P}_{\rm com,2}$. Each sub-problem has different characteristics, therefore, in the following subsections, we provide the solution for each sub-problem using different optimization techniques. 
\subsection{Slave \ac{uav} Optimization}
In this subsection, we aim to optimize the across-track position of the slave \ac{uav}, $\mathbf{q}_2$,  for given resources $\{\mathbf{v}_y, \mathbf{P}_{\mathrm{ com,1}}, \mathbf{P}_{\mathrm{ com,2}}\}$ and {given} master \ac{uav} position $\mathbf{q}_1$. Therefore, problem $\mathrm{(P.1)}$ reduces to sub-problem $\mathrm{(P.1.a)}$ given by:
\begin{alignat*}{2} 
	&(\mathrm{P.1.a}):\max_{\mathbf{q}_2} \hspace{3mm}  C_N(\mathbf{q}_1,\mathbf{q}_2,\mathbf{v}_y)   & \qquad&  \\
	\text{s.t.} \hspace{3mm} &\mathrm{C3-C9,C14,C15}, &       &      \\& \mathrm{C1: } \; z_{\mathrm{min}} \leq z_2 \leq z_{\mathrm{ max}},            &      &  \\ & \mathrm{C11}: R_{2,n}(\mathbf{q}_2,\mathbf{P}_{\mathrm{com},2},\mathbf{v}_y) \geq R_{\mathrm{min},2} (\mathbf{q}_2), \forall n.      & & 
\end{alignat*}
Sub-problem $(\mathrm{P.1.a})$ is a non-smooth non-convex optimization problem due to its objective function and constraints $\mathrm{ C5}$, $\mathrm{C7-C9}$, and $\mathrm{C11}$. The {presence} of trigonometric functions in constraints $\mathrm{ C7}$, $\mathrm{ C8}$, and $\mathrm{ C9}$, as well as the non-smoothness of the objective function \ac{wrt} $\mathbf{q}_2$ make it very difficult to solve problem  $(\mathrm{P.1.a})$ using classical optimization methods that are generally based on  gradient {descent}. Nevertheless, we provide a low-complexity close-to-optimal  solution based on stochastic optimization \cite{stochastic_optimization}. { Motivated by its success in solving complex non-smooth optimization problems and its {fast} convergence \cite{PSO_overview}}, we use {particle swarm optimization (PSO)}, which is a population-based stochastic optimization {framework} that mimics the intelligent collective behavior of animals, such as birds \cite{pso_origin}. To solve sub-problem $\mathrm{(P.1.a)}$, a population of $D$ particles is employed to  explore the \ac{2d} across-track plane. In iteration $k$ of the \ac{pso} algorithm,  particle $d$, denoted by $\mathbf{p}_d^{(k)} \in \mathbb{R}^2, \forall d,$ has {two} dimensions corresponding to {the} $x$- and $z$-coordinates of the slave \ac{uav} as follows:
\begin{equation}\label{eq:pso_particle}
	\mathbf{p}_d^{(k)}=(x_{2,d}^{(k)},z_{2,d}^{(k)})^T, \forall d,
\end{equation}
where $x_{2,d}^{(k)}$ and $z_{2,d}^{(k)}$ correspond to the $x$- and {$z$-positions} of particle $d$ in iteration $k$, respectively. Furthermore,  the velocity vector {of particle $d$ in iteration $k$ }is denoted by $\mathbf{v}_{\mathrm{PSO}, d}^{(k)} \in \mathbb{R}^2$ and given by: 
\begin{equation}\label{eq:PSO_velocity}
	\mathbf{v}_{\mathrm{PSO}, d}^{(k)}=\big(v_{\mathrm{PSO}, d}^{(k)}[1],v_{\mathrm{PSO}, d}^{(k)}[2]\big)^T, \forall d,
\end{equation}
where $v_{\mathrm{PSO}, d}^{(k)}[1]$ and $v_{\mathrm{PSO}, d}^{(k)}[2]$ are the $x$- and {$z$-components} of the particle velocity vector, respectively.  In iteration $k \geq 2$, the  velocity of each particle is affected by its previous velocity, given by $\mathbf{v}_{\mathrm{PSO}, d}^{(k-1)}$, its previous local experience, given by its previous local best-known position, denoted by $\mathbf{p}_{\mathrm{best}, d}^{(k-1)}\in \mathbb{R}^{2}$, and its global experience, given by the best-known position, denoted by $\mathbf{p}_{\mathrm{best}}^{(k-1)}\in \mathbb{R}^{2}$, as shown below \cite{pso_parameter}:
\ifonecolumn
\begin{equation}\label{eq:pso_velocity_update}
	\mathbf{v}_{\mathrm{PSO}, d}^{(k)}=  \stackrel{\text{Inertial term}}{\overbrace{w(k-1) \mathbf{v}_{\mathrm{PSO}, d}^{(k-1)}}} +\stackrel{\text{Cognitive term}}{\overbrace{ c_1 \mathcal{R}_1 (\mathbf{p}_{\mathrm{best}, d}^{(k-1)}-\mathbf{p}_d^{(k-1)} )}} +\stackrel{\text{Social term}}{\overbrace{ c_2  \mathcal{R}_2 (\mathbf{p}_{\mathrm{best}}^{(k-1)}-\mathbf{p}_d^{(k-1)})}}, \forall k \geq2,
\end{equation} 
\else
\begin{align}\label{eq:pso_velocity_update}
	\mathbf{v}_{\mathrm{PSO}, d}^{(k)}&=  \stackrel{\text{Inertial term}}{\overbrace{w(k-1) \mathbf{v}_{\mathrm{PSO}, d}^{(k-1)}}} +\stackrel{\text{Cognitive term}}{\overbrace{ c_1 \mathcal{R}_1 (\mathbf{p}_{\mathrm{best}, d}^{(k-1)}-\mathbf{p}_d^{(k-1)} )}} + \notag \\ 
	&	\stackrel{\text{Social term}}{\overbrace{ c_2  \mathcal{R}_2 (\mathbf{p}_{\mathrm{best}}^{(k-1)}-\mathbf{p}_d^{(k-1)})}}, \forall k \geq2,
\end{align} 
\fi
where $c_1$ and $c_2$  denote the cognitive and social learning {factors}, respectively, $\mathcal{R}_1$ and $\mathcal{R}_2$ are random variables that are uniformly distributed in $[0,1]$, and $w(k-1)\in [0,1]$ is the inertial weight in iteration $k-1$. Different  methods to adjust the inertial weight were proposed in the literature  \cite{decay}. In this work, we employ a linearly decaying inertial weight. Moreover, {the} initial velocity vectors, denoted by $\mathbf{v}_{\mathrm{PSO}, d}^{(1)}, \forall d$, are uniformly distributed in $[0,v_{\rm PSO}^{\rm max}]\times[0,v_{\rm PSO}^{\rm max}]$, where $v_{\rm PSO}^{\rm max}$ is the  maximum particle velocity. Note that {the} updates in (\ref{eq:pso_velocity_update}) might {cause} particles {to move} outside of the feasible search space. To prevent {this from happening} and {to} speed up the convergence of the \ac{pso} algorithm, we define boundaries for the search space based on constraints $\mathrm{C1}$ and $\mathrm{C4}$ by {defining} a reflecting wall  \cite{boundary} as follows: 
\ifonecolumn
\begin{equation}\label{eq:reflecting_wall}
	\begin{dcases}
		v_{\mathrm{PSO}, d}^{(k)} [1] &=	- v_{\mathrm{PSO}, d}^{(k)}[1] \hspace{5mm}\text{ if } x_{2,d}^{(k)} + v_{\mathrm{PSO}, d}^{(k)} [1]> x_t, \\
		v_{\mathrm{PSO}, d}^{(k)} [2] &=	- v_{\mathrm{PSO}, d}^{(k)} [2] \hspace{5mm}\text{ if } z_{2,d}^{(k)} +v_{\mathrm{PSO}, d}^{(k)} [2]> z_{\mathrm{max}} \text{ or } z_{2,d}^{(k)}+v_{\mathrm{PSO}, d}^{(k)} [2] < z_{\mathrm{min}},
	\end{dcases}, \forall d, \forall k \geq 2.
\end{equation}
\else
\begin{equation}\label{eq:reflecting_wall}	\begin{adjustbox}{width=\columnwidth} $
		\begin{dcases}
			v_{\mathrm{PSO}, d}^{(k)} [1] &=	- v_{\mathrm{PSO}, d}^{(k)}[1] \hspace{5mm}\text{ if } x_{2,d}^{(k)} > x_t, \\
			v_{\mathrm{PSO}, d}^{(k)} [2] &=	- v_{\mathrm{PSO}, d}^{(k)} [2] \hspace{5mm}\text{ if } z_{2,d}^{(k)} > z_{\mathrm{max}} \text{ or } z_{2,d}^{(k)} < z_{\mathrm{min}},
		\end{dcases}, \forall d, \forall k \geq 2. $
	\end{adjustbox}
\end{equation}
\fi Next, in iteration $k$, the particle position is updated as follows: 
\begin{equation} \label{eq:pso_update}
	\mathbf{p}_d^{(k+1)}=\mathbf{p}_d^{(k)}+\mathbf{v}_{\mathrm{PSO}, d}^{(k)}, \forall d, \forall k \geq 1,
\end{equation}
where  initial particles $\mathbf{p}_d^{(1)},\forall d,$ are uniformly distributed in $[x_t-O,x_t]\times[z_{\rm min},z_{\rm max}]$ according to $\mathrm{C1}$ and $\mathrm{C4}$, where $O$ is set large enough, {such that the two drones are placed sufficiently behind $x_t$, see constraint $\mathrm{C4}$}. 
The {remaining constraints} of sub-problem $(\mathrm{P.1.a})$ as well as the objective function are used to construct a  fitness function  necessary to evaluate a given particle of the population. { While the \ac{pso} algorithm was originally  proposed {to solve} unconstrained optimization {problems} \cite{pso_origin}, several variations have been proposed to accommodate {constraints}, e.g.,  \cite{pso_coevolution,pso_parameter,pso_ga}.} In this work, we {adopt a} non-parameterized \ac{pso} algorithm, which was proposed in \cite{pso_ga}. In particular, {our} fitness function $\mathcal{F}$ is given by: 
\ifonecolumn
\begin{equation}\label{eq:fitness}
	\mathcal{F}(\mathbf{p}_d^{(k)})=\begin{dcases}
		C_N(\mathbf{q}_1,\mathbf{p}_d^{(k)},\mathbf{v}_y),\hspace{5mm}&\text{ if } \mathbf{p}_d^{(k)} \in S,\\
		\min\limits_{k,d}\left(C_N(\mathbf{q}_1,\mathbf{p}_d^{(k)},\mathbf{v}_y)\right)	+ \sum\limits_{l \in \mathcal{L}}g_l(\mathbf{p}_d^{(k)}), \hspace{5mm}&\text{ if } \mathbf{p}_d^{(k)} \notin S,\\
	\end{dcases} \hspace{5mm}\forall d,
\end{equation}
\else
\begin{equation}\label{eq:fitness}	\begin{adjustbox}{width=\columnwidth} $
	\mathcal{F}(\mathbf{p}_d^{(k)})=\begin{dcases}
		C_N(\mathbf{q}_1,\mathbf{p}_d^{(k)},\mathbf{v}_y), &\text{ if } \mathbf{p}_d^{(k)} \in S,\\ 
		\min\limits_{k,d}\left(C_N(\mathbf{q}_1,\mathbf{p}_d^{(k)},\mathbf{v}_y)\right)	- \sum\limits_{l \in \mathcal{L}}g_l(\mathbf{p}_d^{(k)}), &\text{ if } \mathbf{p}_d^{(k)} \notin S,\\
	\end{dcases} \forall d,$
\end{adjustbox}
\end{equation}
\fi
where $S$ is the feasible set for sub-problem $\mathrm{(P.1.a)}$, {which is defined by constraints $\mathrm{C3}$, $\mathrm{C5-C9}$, $\mathrm{C11}$, $\mathrm{C14}$, and  $\mathrm{C15}$, and functions $g_l, l \in \mathcal{L}=\{3,5,6,7,8,9,11,14\}$, which are used as metrics to quantify constraint violations as follows:}
\ifonecolumn
\begin{align}
	&g_3(\mathbf{p}_d^{(k)})= \left[  r_2(\mathbf{p}_d^{(k)})-r_1(\mathbf{q}_1)\right]^+, \forall d,\label{eq:g3}\\
	&g_5(\mathbf{p}_d^{(k)})= \left[ b_{\mathrm{min}}-b(\mathbf{q}_1,\mathbf{p}_d^{(k)})\right]^+, \forall d,\\
	&g_6(\mathbf{p}_d^{(k)})= \left[ \gamma_{\mathrm{SNR}}^{\mathrm{min}}-\sum\limits_{n=1}^{N} \gamma_{\mathrm{SNR},n}(\mathbf{q}_1,\mathbf{p}_d^{(k)},\mathbf{v}_y)\right]^+, \forall d, \\
	&g_7(\mathbf{p}_d^{(k)})= \left[ \gamma_{\mathrm{Rg}}^{\mathrm{min}}-\gamma_{\mathrm{Rg}}(\mathbf{p}_d^{(k)})\right]^+, \forall d,   \\
	&g_8(\mathbf{p}_d^{(k)})= \left[  h_{\mathrm{amb}}^{\mathrm{min}} -  h_{\mathrm{amb}}(\mathbf{q}_1,\mathbf{p}_d^{(k)})\right]^+, \forall d,   \\
	&g_9(\mathbf{p}_d^{(k)})= \left[   \Delta h_{90\%}^{\rm worst}(\mathbf{q}_1,\mathbf{p}_d^{(k)})-\Delta h^{\mathrm{max}}\right]^+, \forall d,   \\
	&g_{11}(\mathbf{p}_d^{(k)})= \left[ R_{\mathrm{min},2}(\mathbf{p}_d^{(k)}) -\sum\limits_{n=1}^{N} R_{2,n}(\mathbf{p}_d^{(k)},\mathbf{P}_{\mathrm{com},2},\mathbf{v}_y)\right]^+, \forall d,\label{eq:g11}\\
	&g_{14}(\mathbf{p}_d^{(k)})= \left[ \left(\theta_{\rm min} - \theta_2(\mathbf{p}_d^{(k)})\right)+\left(\theta_2(\mathbf{p}_d^{(k)})-\theta_{\rm max}\right)\right]^+, \forall d.\label{eq:g14}	
\end{align} 
\else
\begin{align}
	&g_3(\mathbf{p}_d^{(k)})= \left[  r_2(\mathbf{p}_d^{(k)})-r_1(\mathbf{q}_1)\right]^+, \forall d,\label{eq:g3}\\
	&g_5(\mathbf{p}_d^{(k)})= \left[ b_{\mathrm{min}}-b(\mathbf{q}_1,\mathbf{p}_d^{(k)})\right]^+, \forall d,\\
	&\begin{adjustbox}{width=0.85\columnwidth} $ g_6(\mathbf{p}_d^{(k)})= \sum\limits_{n=1}^{N}\left[ \gamma_{\mathrm{SNR}}^{\mathrm{min}}- \gamma_{\mathrm{SNR},n}(\mathbf{q}_1,\mathbf{p}_d^{(k)},\mathbf{v}_y)\right]^+, \forall d,$
	\end{adjustbox} \\
	&g_7(\mathbf{p}_d^{(k)})= \left[ \gamma_{\mathrm{Rg}}^{\mathrm{min}}-\gamma_{\mathrm{Rg}}(\mathbf{p}_d^{(k)})\right]^+, \forall d,   \\
	&g_8(\mathbf{p}_d^{(k)})= \left[  h_{\mathrm{amb}}^{\mathrm{min}} -  h_{\mathrm{amb}}(\mathbf{q}_1,\mathbf{p}_d^{(k)})\right]^+, \forall d,   \\
	&g_9(\mathbf{p}_d^{(k)})= \left[   \Delta h_{90\%}^{\rm worst}(\mathbf{q}_1,\mathbf{p}_d^{(k)})-\Delta h^{\mathrm{max}}\right]^+, \forall d,   \\
	&\begin{adjustbox}{width=0.85\columnwidth} $g_{11}(\mathbf{p}_d^{(k)})= \sum\limits_{n=1}^{N}\left[ R_{\mathrm{min},2}(\mathbf{p}_d^{(k)}) - R_{2,n}(\mathbf{p}_d^{(k)},\mathbf{P}_{\mathrm{com},2},\mathbf{v}_y)\right]^+, \forall d,\label{eq:g11}$
	\end{adjustbox}\\
	&\begin{adjustbox}{width=0.85\columnwidth} $g_{14}(\mathbf{p}_d^{(k)})= \left[ \theta_{\rm min} - \theta_2(\mathbf{p}_d^{(k)})\right]^++\left[\theta_2(\mathbf{p}_d^{(k)})-\theta_{\rm max}\right]^+, \forall d.\label{eq:g14}	$
	\end{adjustbox}
\end{align} 
\fi
\begin{algorithm}
	\caption{Particle Swarm Optimization Algorithm }\label{alg:pso}
	\begin{algorithmic}[1] 
		\label{algorithm1} \State For given $\{\mathbf{q}_1,\mathbf{v}_y,\mathbf{P}_{\mathrm{com},1},\mathbf{P}_{\mathrm{com},2}\}$, set initial iteration number $k=1$, population of $D$ particles $(\mathbf{p}_1^{(1)}, ..., \mathbf{p}_D^{(1)})$ with initial random velocities $(\mathbf{v}_{\mathrm{PSO}, 1}^{(1)}, ..., \mathbf{v}_{\mathrm{PSO}, D}^{(1)})$, initial best local position for each particle $(\mathbf{p}_{\mathrm{best}, 1}^{(1)}=\mathbf{p}_1^{(1)}, ..., \mathbf{p}_{\mathrm{best}, D}^{(1)}=\mathbf{p}_D^{(1)})$, initial global best swarm position $\mathbf{p}_{\mathrm{best}}^{(1)}=\arg\max\limits_{\mathbf{p}_1^{(1)}, ..., \mathbf{p}_D^{(1)}}(\mathcal{F}(\mathbf{p}_1^{(1)}),...,\mathcal{F}(\mathbf{p}_D^{(1)}))$, and initial inertial weight $w(1)=1$. 
		\State \textbf{repeat}
		\State \hspace{\algorithmicindent}$ $Adjust velocities $\mathbf{v}_{\mathrm{PSO}, d}^{(k)}, \forall d,$ based on (\ref{eq:reflecting_wall}). 
		\State \hspace{\algorithmicindent}$ $Update positions  $\mathbf{p}_d^{(k+1)} = \mathbf{p}_d^{(k)}+\mathbf{v}_{\mathrm{PSO}, d}^{(k)}, \forall d$ 
		\State\hspace{\algorithmicindent}$ $Update population velocities $ \mathbf{v}_{\mathrm{PSO}, d}^{(k+1)}=  w(k) \mathbf{v}_{\mathrm{PSO}, d}^{(k)} +$ \Statex\hspace{\algorithmicindent}$ $$c_1 \mathcal{R}_1 (\mathbf{p}_{\mathrm{best}, d}^{(k)}-\mathbf{p}_d^{(k)} ) + c_2 \mathcal{R}_2 (\mathbf{p}_{\mathrm{best}}^{(k)}-\mathbf{p}_d^{(k)}), \forall d$
		\State\hspace{\algorithmicindent}$ $Update  inertial weight  $ w(k+1)=1-\frac{k}{M_1}$
		
		\State\hspace{\algorithmicindent}$ $Update  local best known position $\mathbf{p}_{\mathrm{best}, d}^{(k+1)}$ for each \State\hspace{\algorithmicindent}$ $particle and the swarm’s best known position\Statex\hspace{\algorithmicindent}$ $$\mathbf{p}_{\mathrm{best}}^{(k+1)}$ based on fitness function $\mathcal{F}$ {in (\ref{eq:fitness})}
		\State\hspace{\algorithmicindent}$ $Set iteration number $k=k+1$
		\State \textbf{until} convergence
		\State \textbf{return} solution $ \mathbf{q}_2=\mathbf{p}_{\rm best}^{(k-1)}$
	\end{algorithmic}
\end{algorithm}
The proposed procedure to solve sub-problem $\mathrm{(P.1.a)}$ is  summarized in \textbf{Algorithm} \ref{alg:pso}. According to our simulations,    \textbf{Algorithm} \ref{alg:pso} converges to {a} close-to-optimal solution, which {can be} assessed {using an} upper bound {for} $\mathrm{(P.1.a)}$;   $C_N(\mathbf{q}_1,\mathbf{q}_2,\mathbf{v}_y) \leq  \big( S_{\rm far}(\mathbf{q}_1) - S_{\rm near}(\mathbf{q}_1)\big) \sum\limits_{n=1}^{N-1}v_y[n]\delta_t$. {Next, we analyze the worst-case  computational time complexity of \textbf{Algorithm} 1. To this end, let} $M_1$ be the maximum number of iterations for \textbf{Algorithm} \ref{alg:pso}. Then, to evaluate the fitness function for each particle, $2N+6$  computations are required, see (\ref{eq:g3})-(\ref{eq:g14}). Therefore, the initialization of \textbf{Algorithm} \ref{alg:pso} requires $1+3D+D(2N+6)$ computations.  {Consequently}, the overall worst-case time complexity of \textbf{Algorithm} 1 is given by $\mathcal{O}\big(1+3D+D(2N+6)+M_1 D(4+2N+6)\big)$, which is equivalent to $\mathcal{O}(M_1 D N)$. 
\subsection{Master \ac{uav} Position Optimization}
\sloppy Next, we focus on  optimizing the master \ac{uav} position in the across-track plane, $\mathbf{q}_1$, for given  $\{\mathbf{v}_y, \mathbf{P}_{\mathrm{ com,1}}, \mathbf{P}_{\mathrm{ com,2}},\mathbf{q}_2\}$. {In this case}, problem $\mathrm{(P.1)}$ reduces to sub-problem $\mathrm{(P.1.b)}$ given by: 
\begin{alignat*}{2} 
	&(\mathrm{P.1.b}):\max_{\mathbf{q}_1} \hspace{3mm}  C_N(\mathbf{q}_1,\mathbf{q}_2,\mathbf{v}_y)   & \qquad&  \\
	\text{s.t.} \hspace{3mm} &\mathrm{C2}, \mathrm{C3},  \mathrm{C5}, \mathrm{C6}, \mathrm{C8}, \mathrm{C9},                &      & \\&  \mathrm{C1}:  \; z_{\mathrm{min}} \leq z_1 \leq z_{\mathrm{ max}},    &      &   \\
	& \mathrm{C11}: R_{1,n}(\mathbf{q}_1,\mathbf{P}_{\mathrm{com},1},\mathbf{v}_y) \geq R_{\mathrm{min},1} (\mathbf{q}_1), \forall n.       & &  
\end{alignat*}
Sub-problem $(\mathrm{P.1.b})$ is non-convex due to its objective function and constraints $\mathrm{C5}$ and $\mathrm{C11}$.  {Thus}, finding the global optimal solution to sub-problem  $(\mathrm{P.1.b})$ is challenging. Yet, we provide the optimal solution to sub-problem $(\mathrm{P.1.b})$ based on monotonic optimization theory. To this end, we start by transforming  $(\mathrm{P.1.b})$ {into} the canonical form of a monotonic optimization problem \cite{2}. 
\begin{proposition}\label{prop:equivalence}
{Sub-problem $(\mathrm{P.1.b})$ is equivalent to the following  monotonic optimization problem}: 
\begin{alignat*}{2} 
	&(\overline{\mathrm{P.1.b}}):\max_{z_1,t} \hspace{3mm}  C_N(\mathbf{q}_1,\mathbf{q}_2,\mathbf{v}_y)   & \qquad&  \\
	\text{s.t.} \hspace{3mm} &\mathrm{C6},                &      & \\&  \mathrm{\overline{C1a}}:  \;  z_1-Z_{\mathrm{ max}} \leq 0,    &      & \\
	&  \mathrm{\overline{C1b}}:  \; z_1- Z_{\mathrm{min}}\geq 0,    &      & \\
	&\overline{\mathrm{ C11a}}: 	{a_1(z_1)}+t\leq {a_1(z_{\rm max})},  &      & \\ 
	&\overline{\mathrm{ C11b}}:    	{a_2(z_1)}+t \geq 	{a_1(z_{\rm max})},  &      & \\
	&\overline{\mathrm{ C11c}}:    0 \leq t \leq {a_1(z_{\rm max})}-{a_1(0)},  &      & 
\end{alignat*}
where $t \in \mathbb{R}$ is an auxiliary optimization variable,   $Z_{\mathrm{min}}$ is given by:
\ifonecolumn
\begin{equation}\label{eq:Zmin}
		Z_{\mathrm{min}}=
		\begin{dcases}
			\cos(\theta_1)\max\Big\{\frac{z_{\mathrm{min}}}{\cos(\theta_1)},r_2(\mathbf{q}_2),\frac{b_{\bot}(\mathbf{q}_2)h_{\mathrm{min}}}{\lambda \sin(\theta_1)}\Big\},& \text{if } b_{\mathrm{min}}\leq b_{\bot}(\mathbf{q}_2),\vspace{5 mm}\\
			\!\begin{aligned} &\cos(\theta_1)\max\Big\{\frac{z_{\mathrm{min}}}{\cos(\theta_1)},r_2(\mathbf{q}_2),\frac{b_{\bot}(\mathbf{q}_2)h_{\mathrm{min}}}{\lambda \sin(\theta_1)},r_2(\mathbf{q}_2)\cos\big(\theta_1-\theta_2(\mathbf{q}_2)\big) +  \\ & \sqrt{b^2_{\mathrm{min}}-b^2_{\bot}(\mathbf{q}_2)}\Big\},\end{aligned}&\text{otherwise, }
		\end{dcases} 
\end{equation}
\else
\begin{equation}\label{eq:Zmin}
	\begin{adjustbox}{width=\columnwidth} $
	\begin{dcases}
		\cos(\theta_1)\max\Big\{\frac{z_{\mathrm{min}}}{\cos(\theta_1)},r_2(\mathbf{q}_2),\frac{b_{\bot}(\mathbf{q}_2)h_{\mathrm{min}}}{\lambda \sin(\theta_1)}\Big\},& \text{if } b_{\mathrm{min}}\leq b_{\bot}(\mathbf{q}_2),\vspace{5 mm}\\
	\!\begin{aligned} &\cos(\theta_1)\max\Big\{\frac{z_{\mathrm{min}}}{\cos(\theta_1)},r_2(\mathbf{q}_2),\frac{b_{\bot}(\mathbf{q}_2)h_{\mathrm{min}}}{\lambda \sin(\theta_1)}, \\ &r_2(\mathbf{q}_2)\cos\big(\theta_1-\theta_2(\mathbf{q}_2)\big) +  \sqrt{b^2_{\mathrm{min}}-b^2_{\bot}(\mathbf{q}_2)}\Big\},\end{aligned}&\text{otherwise, }
	\end{dcases} $
\end{adjustbox}
\end{equation}
\fi
$Z_{\mathrm{max}}$ is given by: 
\ifonecolumn
\begin{equation} \label{eq:Zmax}
Z_{\mathrm{max}}=\begin{dcases}
	\cos(\theta_1)\max\Big\{\frac{z_{\mathrm{max}}}{\cos(\theta_1)},\frac{2\pi b_{\bot}(\mathbf{q}_2)\Delta h^{\mathrm{max}} }{\lambda \sin(\theta_1)\Delta\phi_{90\%}(\gamma_{\mathrm{SNR}}^{\mathrm{min}}\gamma_{\mathrm{Rg}}^{\mathrm{min}}\gamma_{\mathrm{other}})}\Big\},	&\text{if } b_{\mathrm{min}}\leq b_{\bot}(\mathbf{q}_2),\\
		\!\begin{aligned}
			&\cos(\theta_1)\max\Big\{\frac{z_{\mathrm{max}}}{\cos(\theta_1)},r_2(\mathbf{q}_2)\cos\big(\theta_1-\theta_2(\mathbf{q}_2)\big) -\sqrt{b^2_{\mathrm{min}}-b^2_{\bot}(\mathbf{q}_2)},\\ &\frac{2\pi  b_{\bot}(\mathbf{q}_2)\Delta h^{\mathrm{max}}}{\lambda \sin(\theta_1)\Delta\phi_{90\%}(\gamma_{\mathrm{SNR}}^{\mathrm{min}}\gamma_{\mathrm{Rg}}^{\mathrm{min}}\gamma_{\mathrm{other}})}\Big\},\end{aligned}&\text{otherwise. }
	\end{dcases}
\end{equation}
\else
\begin{equation} \label{eq:Zmax}
		\begin{adjustbox}{width=\columnwidth} $
\begin{dcases}	\cos(\theta_1)\max\Big\{\frac{z_{\mathrm{max}}}{\cos(\theta_1)},\frac{2\pi b_{\bot}(\mathbf{q}_2)\Delta h^{\mathrm{max}} }{\lambda \sin(\theta_1)\Delta\phi_{90\%}(\gamma_{\mathrm{SNR}}^{\mathrm{min}}\gamma_{\mathrm{Rg}}^{\mathrm{min}})}\Big\},	&\text{if } b_{\mathrm{min}}\leq b_{\bot}(\mathbf{q}_2),\\
		\!\begin{aligned}
			&\cos(\theta_1)\max\Big\{\frac{z_{\mathrm{max}}}{\cos(\theta_1)},r_2(\mathbf{q}_2)\cos\big(\theta_1-\theta_2(\mathbf{q}_2)\big) -\\ &\sqrt{b^2_{\mathrm{min}}-b^2_{\bot}(\mathbf{q}_2)},\frac{2\pi  b_{\bot}(\mathbf{q}_2)\Delta h^{\mathrm{max}}}{\lambda \sin(\theta_1)\Delta\phi_{90\%}(\gamma_{\mathrm{SNR}}^{\mathrm{min}}\gamma_{\mathrm{Rg}}^{\mathrm{min}}\gamma_{\mathrm{other}})}\Big\},\end{aligned}&\text{otherwise. }
	\end{dcases}$
\end{adjustbox}
\end{equation}
\fi
{Furthermore, }function $a_1$ is given by: 
\ifonecolumn
\begin{align}a_1(z_1)&=\Big(2^{A_1z_1  + A_2}-1\Big)\Big( z_1^2 (\tan^2(\theta_1)+1)+(g_x-x_t)^2+\notag\\g_z^2+ &2z_1\tan(\theta_1) [g_x-x_t]^+\Big)+\max\limits_{n}\Big(P_{\mathrm{com},1}[n] {\beta_{c,1}} \Big) +\notag \\\label{eq:g1} &\max\limits_{1 \leq n \leq N}\Big\{\big(2^{A_1z_1  + A_2}-1\big)(y[n]-g_y)^2 \Big\},
\end{align}
\else
\begin{align}
	a_1(z_1)&=\Big(2^{A_1z_1  + A_2}-1\Big)\Big( z_1^2 (\tan^2(\theta_1)+1)+(g_x-x_t)^2+\notag\\g_z^2+ &2z_1\tan(\theta_1) [g_x-x_t]^+\Big)+\max\limits_{n}\Big(P_{\mathrm{com},1}[n] {\beta_{c,1}} \Big) +\notag \\\label{eq:g1} &\max\limits_{1 \leq n \leq N}\Big\{\big(2^{A_1z_1  + A_2}-1\big)(y[n]-g_y)^2 \Big\},
\end{align}
\fi
where $ A_1=  \frac{2 n_B  B_{\mathrm{Rg}}\mathrm{PRF}}{c B_{c,1}}\left( \frac{1}{\cos \left(\theta_1+ \frac{\Theta_{\mathrm{3dB}}}{2} \right)}  -\frac{1}{\cos \left(\theta_1- \frac{\Theta_{\mathrm{3dB}}}{2} \right)} \right)$, $A_2= \frac{n_B  B_{\mathrm{Rg}}\mathrm{PRF}\tau_p}{B_{c,1}}$, and function $a_2$ is given by: 
\begin{align}
	{a_2(z_1)}=\Big(2^{A_1z_1  + A_2}-1\Big)\Big( 2z_1g_z+2z_1\tan(\theta_1) [x_t-g_x]^+\Big) 
\end{align}
.\end{proposition}
\begin{proof}
{Please refer to Appendix \ref{app:monotonic}.}
\end{proof}
{According to \textbf{Proposition} 1, sub-problem $(\mathrm{P.1.b})$ can be written in the canonical form of a monotonic optimization problem given by $(\overline{\mathrm{P.1.b}})$,}
\begin{figure*}[] 
	\centering
	\begin{tabular}{cccc}
		\includegraphics[width=0.25\linewidth]{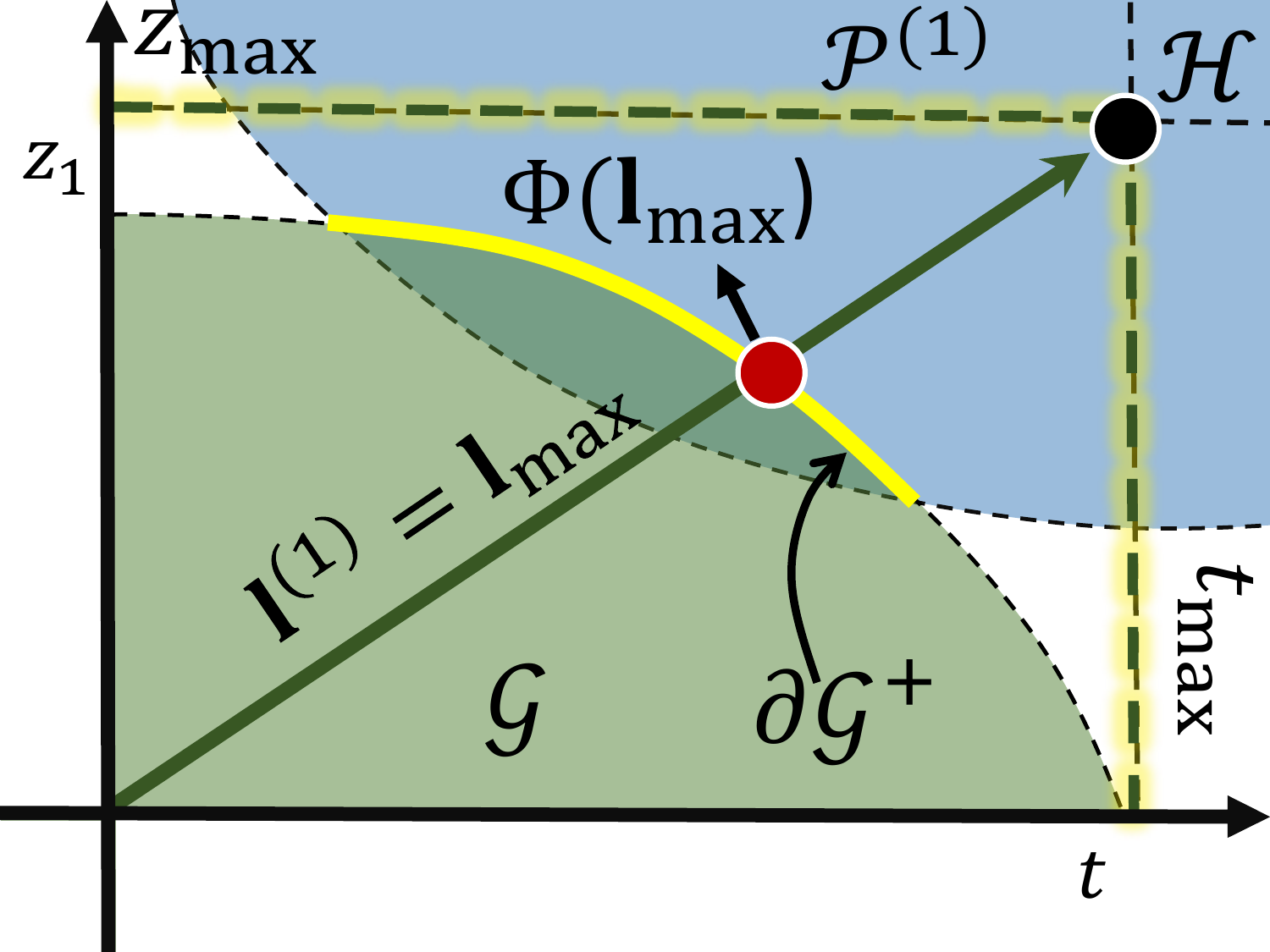}  &
		\includegraphics[width=0.25\linewidth]{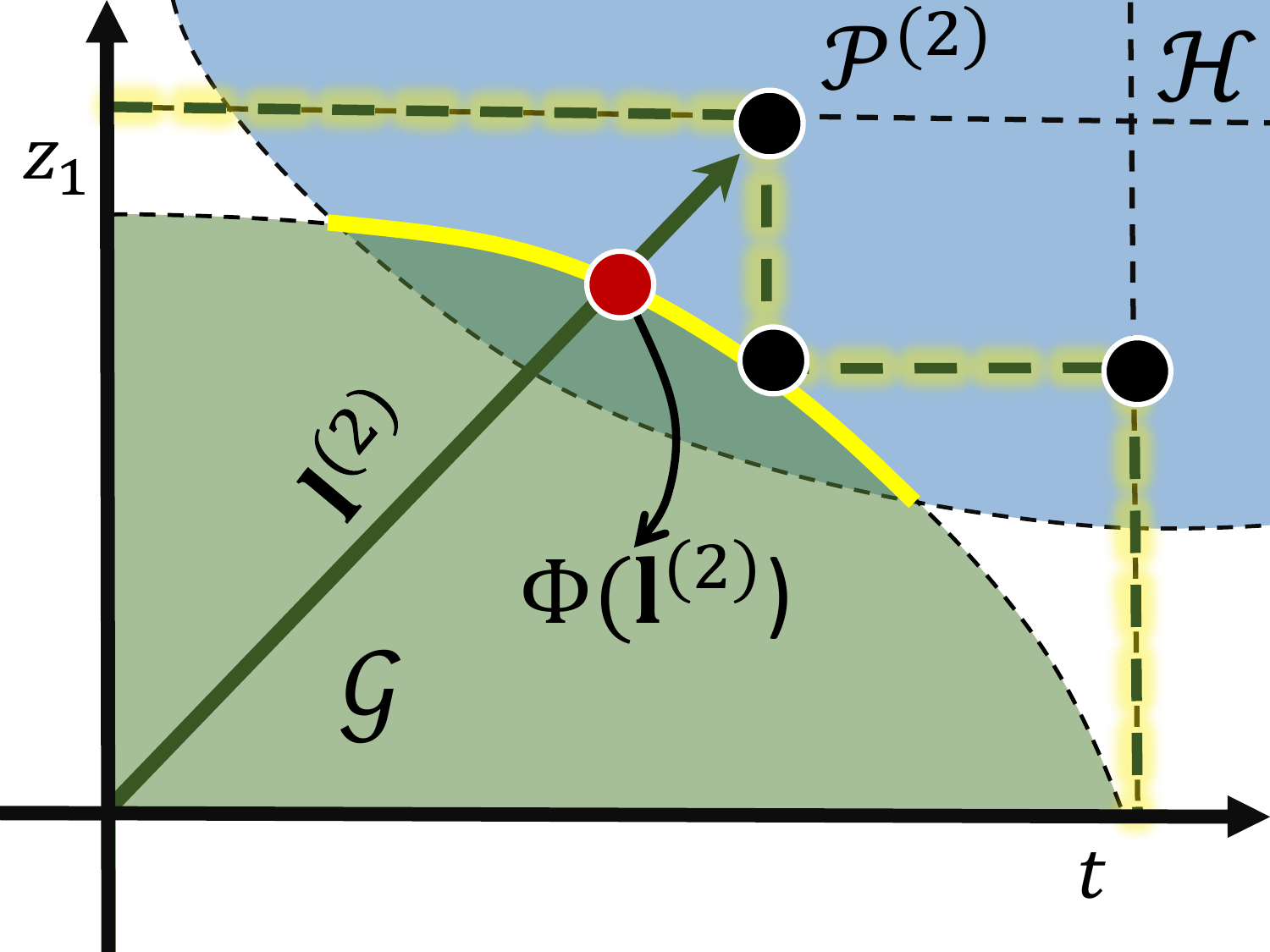}&
		\includegraphics[width=0.25\linewidth]{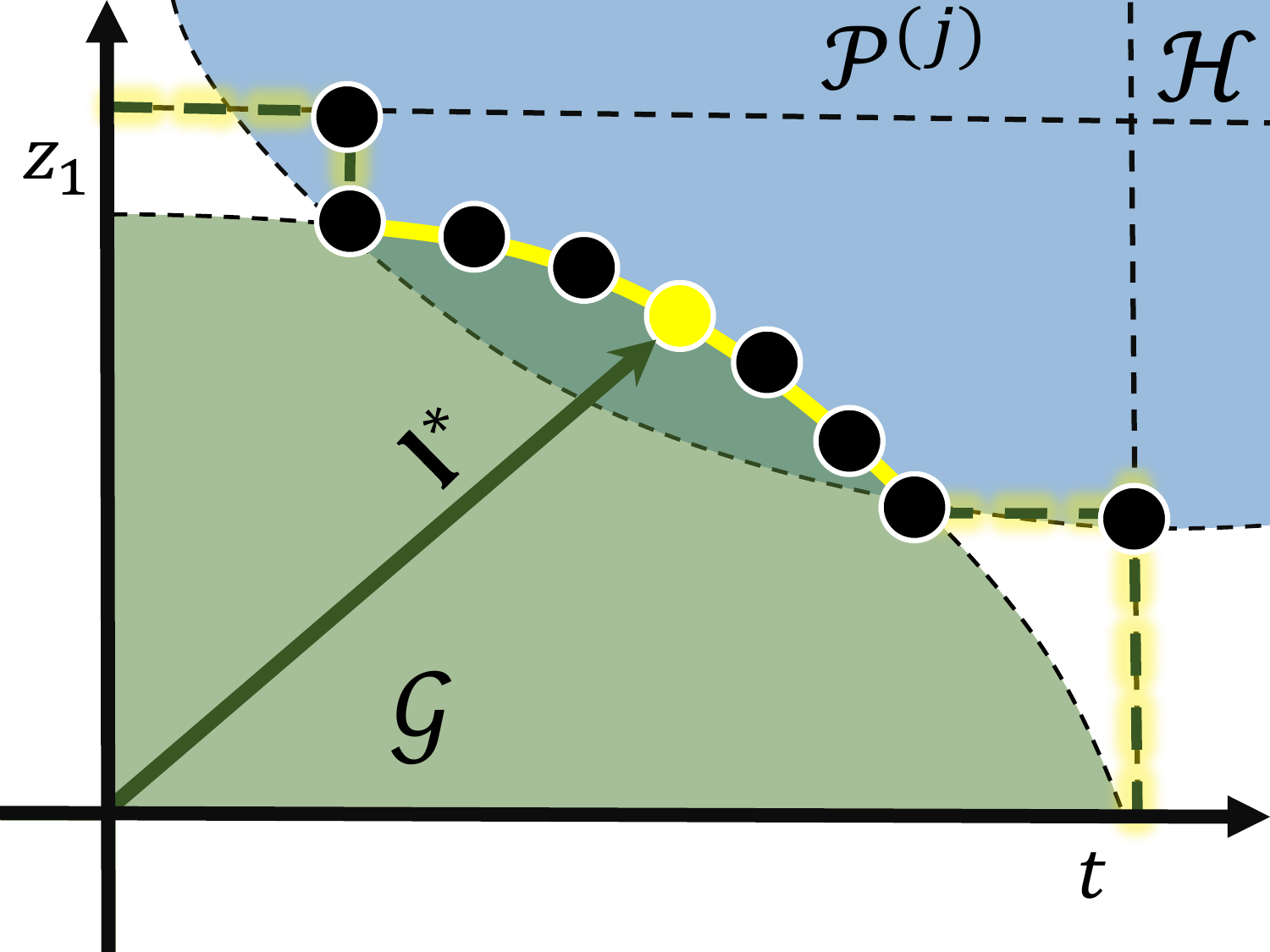} \\
		{\small(a)}& {\small(b)}& {\small(c)}
	\end{tabular}
	\caption{Illustration of polyblock outer approximation algorithm. (a) Feasible set for a monotonic maximization problem where the normal and conormal sets are denoted by $\mathcal{G}$ and $\mathcal{H}$, respectively. Vertex $\mathbf{l}_{\rm max}$ is initially selected as {the} best vertex and its projection {onto} the normal set, denoted by $\Phi(\mathbf{l}_{\rm max})$, is calculated. (b)  {A new} polyblock is constructed and {the} next best vertex $\mathbf{l}^{(2)}=\argmax_{\mathbf{l} \in \mathcal{T}^{(2)}}\{ C_N(\mathbf{l},\mathbf{q}_2, \mathbf{v}_y)\}$ is selected for projection. (c)  {Illustration} of a tighter polyblock obtained after several iterations and {the}  $\epsilon_1$-optimal solution $\mathbf{l}^*$ is located.}
	\label{fig:MO}
\end{figure*}  
where the feasible set is the intersection of normal set $\mathcal{G}$, spanned by constraints ($\overline{\mathrm{C1a}}, \mathrm{C6}, \overline{\mathrm{ C11a}},\overline{\mathrm{ C11c}}$), and conormal set $\mathcal{H}$, spanned by constraints $\rm (\overline{C1b}, \overline{C11b})$. The optimal solution of sub-problem ${(\overline{\mathrm{P.1.b}}})$ lies at the upper boundary of normal set $\mathcal{G}$, denoted by $\partial^+\mathcal{G}$ \cite{2}, see Figure \ref{fig:MO}a. This region is not known a priori but the sequential  polyblock approximation algorithm can be used to approach this set from above. 
\ifonecolumn
 \begin{algorithm}
 	\caption{Polyblock Outer Approximation Algorithm } 	\label{alg:polyblock}
 	\begin{algorithmic}[1] 
 		\State Initialization: Set polyblock $\mathcal{P}^{(1)}$ with vertex $\mathbf{l}^{(1)}=\mathbf{l}_{\rm max} = (z_{\rm max},t_{\rm max})^T$, vertex set $\mathcal{T}^{(1)}=\{ \mathbf{l}_{\rm max}\}$, small positive number $\epsilon_1\geq0$, objective function value $ \rm CBV^{(1)}=- \infty$, current solution $\mathbf{s}^{(1)}=[\cdot]$, and iteration number $j=1$.
 			\State Remove from $\mathcal{T}^{(1)}$ vertices $\{\mathbf{l} \in \mathcal{T}^{(1)}| \mathbf{l} \notin \mathcal{H}\}$
 		\State \textbf{while} $\Big(\mathcal{T}^{(j)} \neq \emptyset$ and $\Big|C_N(\mathbf{l}^{(j)},\mathbf{q}_2,\mathbf{v}_y)-\mathrm{CBV}^{(j)}\Big|> \epsilon_1\Big)$
 		\State\hspace{\algorithmicindent} Select $ \mathbf{l}^{(j)}=\argmax\limits_{ \mathbf{l} \in \mathcal{T}^{(j)}}\Big\{C_N(\mathbf{l},\mathbf{q}_2,\mathbf{v}_y) \Big\}$
 		\State\hspace{\algorithmicindent}  Compute the projection $\Phi(\mathbf{l}^{(j)})$ {onto} $\partial \mathcal{G}^+$ using \textbf{Algorithm} \ref{algorithm3}
 		\State \hspace{\algorithmicindent} \textbf{if}  $\Phi(\mathbf{l}^{(j)}) \in \mathcal{H}$ and $C_N(\Phi(\mathbf{l}^{(j)}),\mathbf{q}_2,\mathbf{v}_y)> \mathrm{CBV}^{(j)}$  \State\hspace{\algorithmicindent} \textbf{then}
 		\State\hspace{\algorithmicindent}\hspace{\algorithmicindent} Let the best solution $\textbf{s}^{(j+1)}=\Phi(\mathbf{l}^{(j)})$ and $\mathrm{CBV}^{(j+1)}=C_N(\Phi(\mathbf{l}^{(j)}),\mathbf{q}_2,\mathbf{v}_y)$ 
 		\State\hspace{\algorithmicindent} \textbf{else} 
 		\State\hspace{\algorithmicindent}\hspace{\algorithmicindent} Set $\textbf{s}^{(j+1)}=\textbf{s}^{(j)}$ and $\mathrm{CBV}^{(j+1)}=\mathrm{CBV}^{(j)}$
 		\State\hspace{\algorithmicindent} \textbf{endif}
 		\State\hspace{\algorithmicindent} Set $\mathcal{T}^{(j+1)}=\{\mathcal{T}^{(j)}\setminus \mathbf{l}^{(j)}\} \cup \Big\{\Big(\Phi(\mathbf{l}^{(j)})[1],l^{(j)}[2]\Big)^T\Big\} \cup$  $\Big\{\Big(l^{(j)}[1],\Phi(\mathbf{l}^{(j)})[2]\Big)^T\Big\}$
 		\State \hspace{\algorithmicindent} Remove from $\mathcal{T}^{(j+1)} $ vertices $\{ \mathbf{l} \in\mathcal{T}^{(j+1)}|  \mathbf{l} \notin \mathcal{H}\}$ 
 		\State \hspace{\algorithmicindent} Set $j=j+1$
 		\State \textbf{end while}
 		\State \textbf{return} solution $\mathbf{q}_1=\Big(\underbrace{x_t-s^{(j)}[1] \tan(\theta_1)}_{x_1},\underbrace{s^{(j)}[1]}_{z_1}\Big)^T$ 
 	\end{algorithmic}
 \end{algorithm}
 \fi
 \begin{algorithm}
 	\caption{Bisection Search Algorithm to Compute the Upper Boundary Point } 	\label{algorithm3}
 	\begin{algorithmic}[1] 
 		\State Initialization: Given $\mathbf{l}^{(j)}$, set $ \lambda_{\rm min}=0$, $ \lambda_{\rm max}=1$, and error tolerance $0< {\epsilon_2} \ll 1$.
 		\State \textbf{repeat}
 		\State\hspace{\algorithmicindent}$\lambda= \frac{(\lambda_{\rm max}+\lambda_{\rm min})}{2}$ \Comment{ Bisection middle point}
 		\State\hspace{\algorithmicindent}\textbf{if} $\lambda \mathbf{l}^{(j)} $ is feasible for ($\overline{\mathrm{C1a}}, \mathrm{C6}, \overline{\mathrm{ C11a}},\overline{\mathrm{ C11c}}$)  \textbf{then}$ $ \State\hspace{\algorithmicindent}\hspace{\algorithmicindent}Set $ \lambda_{\rm min}=\lambda$ \Comment{ Search in [$\lambda,\lambda_{\rm max}$]}
 		\State\hspace{\algorithmicindent}\textbf{else}\Comment{If $\lambda \mathbf{l}^{(j)} \notin \mathcal{G}$}
 		\State\hspace{\algorithmicindent}\hspace{\algorithmicindent}Set $\lambda_{\rm max}=\lambda$\Comment{Search in [$\lambda_{\rm min},\lambda$]}\
 		\State\hspace{\algorithmicindent}\textbf{end if}
 		\State \textbf{until} $\lambda_{\rm max}-\lambda_{\rm min}\leq {\epsilon_2}$\Comment{ Convergence if $\lambda_{\rm max}\approx\lambda_{\rm min}$}
 		\State \textbf{return}  $\Phi(\mathbf{l}^{(j)})=\lambda\mathbf{l}^{(j)}$
 	\end{algorithmic}
 \end{algorithm}
 In \textbf{Algorithm} \ref{alg:polyblock}, we provide all steps of the proposed polyblock outer approximation algorithm that generates an $\epsilon_1$-approximate  optimal solution to sub-problem $(\overline{\mathrm{P.1.b}})$. In addition, Figure \ref{fig:MO}  depicts the operation of the polyblock algorithm. In fact,  optimization variables $z_1$ and $t$ of sub-problem $(\overline{\mathrm{P.1.b}})$   are collected in vertex vector $\mathbf{l}=(l[1],l[2])^T=\big(z_1,t\big)^T \in \mathbb{R}^{2}$. Hereinafter, for ease of notation, $C_N(\mathbf{l},\mathbf{q}_2,\mathbf{v}_y)$  refers to $ C_N((x_t-l[1]\tan(\theta_1),l[1])^T,\mathbf{q}_2,\mathbf{v}_y)$. As shown in Figure \ref{fig:MO}a, first, we  initialize a polyblock $\mathcal{P}^{(1)}$ with corresponding vertex set $\mathcal{T}^{(1)}$ containing vertex $\mathbf{l}_{\rm max} =(z_{\rm max},t_{\rm max})^T \in \mathbb{R}^{2}$, where $t_{\rm max}={a_1(z_{\rm max}) -a_1(0)}$. Then, in iteration $j$, the vertex that maximizes the objective function of sub-problem $(\overline{\mathrm{P.1.b}})$ is selected to create the next polyblock, i.e., $\mathbf{l}^{(j)}=\argmax\limits_{\mathbf{l} \in \mathcal{T}^{(j)}}\{ C_N(\mathbf{l},\mathbf{q}_2,\mathbf{v}_y)\}$.   The projection of vertex $\mathbf{l}^{(j)}$ onto the normal set $\mathcal{G}$, denoted by $\Phi(\mathbf{l}^{(j)})$, is calculated and used to construct a tighter polyblock $\mathcal{P}^{(j)}$, see Figure  \ref{fig:MO}b. The new vertex set is  calculated as  $\mathcal{T}^{(j)}=\{\mathcal{T}^{(j)}\setminus \mathbf{l}^{(j)}\} \cup \left\{\Big(\Phi( \mathbf{l}^{(j)})[1],l^{(j)}[2]\Big)^T\right\}\cup \left\{\Big(l^{(j)}[1],\Phi( \mathbf{l}^{(j)})[2]\Big)^T\right\}$. Once a feasible projection (i.e., $\Phi(\mathbf{l}^{(j)}) \in \mathcal{G}\cap\mathcal{H}$) that improves the objective function of sub-problem $(\overline{\mathrm{P.1.b}})$ is located, the corresponding vector $\mathbf{s}^{(j+1)}=\Phi(\mathbf{l}^{(j)})$ and its objective function  $\mathrm{CBV}^{(j+1)}=C_N(\Phi(\mathbf{l}^{(j)}),\mathbf{q}_2,\mathbf{v}_y)$ are recorded.  To speed up convergence and save memory, vertices that are not in set $\mathcal{H}$ are removed in each iteration \cite{1}.
 The same procedure is repeated multiple times, such that a sequence of polyblocks converges from above to the feasible set, i.e., $\mathcal{P}^{(1)}\supset\mathcal{P}^{(2)}\supset...\supset\mathcal{P}^{(j)}\supset \mathcal{G}\cap\mathcal{H}$, see Figure \ref{fig:MO}c. { The termination} criteria of the proposed algorithm  are (i)  $|C_N(\mathbf{l}^{(j)},\mathbf{q}_2,\mathbf{v}_y)-\mathrm{CBV}^{(j)}|\leq \epsilon_1$, for convergence,  where $\epsilon_1$ is the error tolerance, and (ii) $\mathcal{T}^{(j)} = \emptyset$ for a non-feasible problem.\par In general, \textbf{Algorithm} \ref{alg:polyblock} converges to an $\epsilon_1$-optimal solution with an exponential time complexity in the dimension of the vertex \cite{1}. {Although the polyblock outer approximation algorithm is, in general, computationally expensive, especially when the number of variables is large, for the problem at hand, the vertex dimension is only two, therefore, \textbf{Algorithm} \ref{alg:polyblock} obtains the optimal solution with low computational complexity.} In particular, the time complexity of \textbf{Algorithm} \ref{alg:polyblock} stems mainly from the projection performed by \textbf{Algorithm} \ref{algorithm3}. In \textbf{Algorithm} \ref{algorithm3}, a bisection search in the interval $[0,1]$ is {conducted} with precision ${\epsilon_2}$, thus, the corresponding time complexity is $\mathcal{O}({\log_2}(1/{\epsilon_2}))$. Thus, the complexity of \textbf{Algorithm} \ref{alg:polyblock} is $\mathcal{O}(M_2 {\log_2}(1/{\epsilon_2}))$, where $M_2$ is the maximum number of iterations required for convergence.
\subsection{Resource Allocation Optimization}
In this subsection, we aim to optimize $\{\mathbf{v}_y,\mathbf{P}_{\mathrm{com,1}}, \mathbf{P}_{\mathrm{com,2}}\}$ given \ac{uav} formation $\{\mathbf{q}_1,\mathbf{q}_2\}$. Then, problem $\mathrm{(P.1)}$ reduces to sub-problem  $\mathrm{(P.1.c)}$ given by: 
\begin{alignat*}{2} 
	&(\mathrm{P.1.c}):\max_{\mathbf{P}_{\mathrm{com,1}}, \mathbf{P}_{\mathrm{com,2}}, \mathbf{v}_y} \hspace{3mm}  C_N(\mathbf{q}_1,\mathbf{q}_2,\mathbf{v}_y)   & \qquad&  \\
	\text{s.t.} \hspace{3mm}  
	&  \mathrm{C6},\mathrm{C10},\mathrm{C11}, \mathrm{C12}, \mathrm{C13}.   &      &     
\end{alignat*}
Sub-problem $\mathrm{(P.1.c)}$ is difficult to solve due to the non-convex propulsion power model used in constraint  $\mathrm{C12}$. Yet, we provide a low-complexity sub-optimal solution to $\mathrm{(P.1.c)}$ based on { successive convex approximation (SCA)}.  To start with, we introduce a slack vector $\mathbf{u}=(u[1],...,u[N])^T\in \mathbb{R}^{N}$ and replace the  propulsion power {equivalently} by: 
\ifonecolumn
\begin{align}
	\widetilde{P_{\mathrm{prop},n}}(\mathbf{v}_y,\mathbf{u})=  P_0 \left(1+\frac{3v_y^2[n]}{U^2_{\mathrm{tip}}}\right)+u[n]+\frac{1}{2}d_0 \rho s Av_y^3[n], \forall n. \label{eq:appx_prop}
\end{align} 
\else 
\begin{align}\begin{adjustbox}{width=\columnwidth} $
	\widetilde{P_{\mathrm{prop},n}}(\mathbf{v}_y,\mathbf{u})=  P_0 \left(1+\frac{3v_y^2[n]}{U^2_{\mathrm{tip}}}\right)+u[n]+ \frac{1}{2}d_0 \rho s Av_y^3[n],   \forall n. \label{eq:appx_prop}$ \end{adjustbox}
\end{align} 
\fi
{Constraint $\mathrm{C12}$ is denoted by $\widetilde{\mathrm{C12}}$ after replacing $P_{\mathrm{prop},n}$ with  $\widetilde{P_{\mathrm{prop},n}}$ and the {elements of} slack vector $\mathbf{u}$ must satisfy the following additional constraint:}
\begin{equation}
	\mathrm{C12a}:	1+\frac{v^4_y[n]}{4v_0^4}-\Big( \frac{u^2[n]}{P^2_I}+\frac{v^2_y[n]}{2v^2_0} \Big)^2\leq 0, \forall n.
\end{equation}
However, constraint $\mathrm{C12a}$ is still non-convex due to the term $ f(\mathbf{v}_y,\mathbf{u}) =\Big( \frac{u^2[n]}{P^2_I}+\frac{v^2_y[n]}{2v^2_0} \Big)^2$. Therefore, in iteration $i$ of the \ac{sca} algorithm, we approximate function $f$ based on {its} Taylor {series expansion} around point $\{ \mathbf{v}_y^{(i)},\mathbf{u}^{(i)}\} $ as follows: 
\ifonecolumn
\begin{equation}
	\widetilde{f}(\mathbf{v}_y,\mathbf{u})= f(\mathbf{v}_y^{(i)},\mathbf{u}^{(i)})+ \nabla_{\mathbf{v}_y} f(\mathbf{v}_y^{(i)},\mathbf{u}^{(i)})^T (\mathbf{v}_y-\mathbf{v}_y^{(i)})+\nabla_{\mathbf{u}} f(\mathbf{v}_y^{(i)},\mathbf{u}^{(i)})^T (\mathbf{u}-\mathbf{u}^{(i)}),
\end{equation}
\else 
\begin{align}
	\widetilde{f}(\mathbf{v}_y,\mathbf{u})&= f(\mathbf{v}_y^{(i)},\mathbf{u}^{(i)})+ \nabla_{\mathbf{v}_y} f(\mathbf{v}_y^{(i)},\mathbf{u}^{(i)})^T (\mathbf{v}_y-\mathbf{v}_y^{(i)})+\notag \\&\nabla_{\mathbf{u}} f(\mathbf{v}_y^{(i)},\mathbf{u}^{(i)})^T (\mathbf{u}-\mathbf{u}^{(i)}),
\end{align}
\fi
where $\nabla_{\mathbf{v}_y} f(\mathbf{v}_y^{(i)},\mathbf{u}^{(i)})=\frac{\partial f}{\partial  \mathbf{v}_y}(\mathbf{v}_y^{(i)},\mathbf{u}^{(i)})$ and $\nabla_{\mathbf{u}} f(\mathbf{v}_y^{(i)},\mathbf{u}^{(i)})= \frac{\partial f}{\partial  \mathbf{u}}(\mathbf{v}_y^{(i)},\mathbf{u}^{(i)})$. Function $\widetilde{f}(\mathbf{v}_y,\mathbf{u})$ is {inserted} in constraint  $\mathrm{C12a}$ leading to a new convex constraint  denoted by $\widetilde{\mathrm{C12a}}$. The resulting convex optimization problem is given by: 
\begin{alignat*}{2} 
	(\widetilde{\mathrm{P.1.c}})&:\max_{\mathbf{P}_{\mathrm{com,1}}, \mathbf{P}_{\mathrm{com,2}}, \mathbf{v}_y,\mathbf{u}} \hspace{3mm}  C_N(\mathbf{q}_1,\mathbf{q}_2,\mathbf{v}_y)   & \qquad&  \\
	&  	\text{s.t. }  \mathrm{C6},\mathrm{C10},\mathrm{C11},\widetilde{\mathrm{C12}},\widetilde{\mathrm{C12a}},\mathrm{C13}.
\end{alignat*}
\ifonecolumn
\begin{algorithm}[]
	\caption{Successive Convex Approximation Algorithm}\label{alg:sca}
	\begin{algorithmic}[1] 
		\label{algorithm1} \State For fixed $\{\mathbf{q}_1,\mathbf{q}_2\}$, set initial point $\{\mathbf{v}_y^{(1)},\mathbf{u}^{(1)}\}$, iteration index $i=1$, and error tolerance $0< {\epsilon_3} \ll 1$.
		\State \textbf{repeat}  
		\State\hspace{\algorithmicindent}$ $Determine coverage $ C_N(\mathbf{q}_1,\mathbf{q}_2,\mathbf{v}_y), \mathbf{u}, \mathbf{v}_y,  \mathbf{P}_{\mathrm{ com,1}}, $ and $ \mathbf{P}_{\mathrm{ com,2}} $ by solving $\mathrm{ \widetilde{(P.1.c)}}$ around point 	\Statex\hspace{\algorithmicindent}$\{\mathbf{u}^{(i)},\mathbf{v}_y^{(i)}\}$ using CVXPY       
		\State\hspace{\algorithmicindent}$ $Set $ i=i+1,\hspace{1mm} \mathbf{u}^{(i)}= \mathbf{u},\hspace{1mm} \mathbf{v}_y^{(i)}=\mathbf{v}_y$ \Comment{Update solution and iteration index}
		\State \textbf{until} $ \big |\frac{C_N(\mathbf{q}_1,\mathbf{q}_2,\mathbf{v}_y^{(i)})-C_N(\mathbf{q}_1,\mathbf{q}_2,\mathbf{v}_y^{(i-1)})}{C_N(\mathbf{q}_1,\mathbf{q}_2,\mathbf{v}_y^{(i)})}\big|<  {\epsilon_3} $
		\State \textbf{return} solution \{$\mathbf{v}_y, \mathbf{P}_{\mathrm{ com,1}}, \mathbf{P}_{\mathrm{ com,2}}$\}
	\end{algorithmic}
\end{algorithm}
\else 
\begin{algorithm}[]
	\caption{Successive Convex Approximation for $\mathrm{ (P.1.c)}$ }\label{alg:sca}
	\begin{algorithmic}[1] 
		\label{algorithm1} \State For fixed $\{\mathbf{q}_1,\mathbf{q}_2\}$, set initial point $\{\mathbf{v}_y^{(1)},\mathbf{u}^{(1)}\}$, iteration index $i=1$, and error tolerance $0< {\epsilon_3}  \ll 1$.
		\State \textbf{repeat}  
		\State\hspace{\algorithmicindent}$ $Determine coverage $ C_N(\mathbf{q}_1,\mathbf{q}_2,\mathbf{v}_y), \mathbf{u}, \mathbf{v}_y,  \mathbf{P}_{\mathrm{ com,1}}, $ \Statex\hspace{\algorithmicindent}and $ \mathbf{P}_{\mathrm{ com,2}} $ by solving $\mathrm{ \widetilde{(P.1.c)}}$ around point $\{\mathbf{u}^{(i)},\mathbf{v}_y^{(i)}\}$ \Statex\hspace{\algorithmicindent}using CVXPY       
		\State\hspace{\algorithmicindent}$ $Set $ i=i+1,\hspace{1mm} \mathbf{u}^{(i)}= \mathbf{u},\hspace{1mm} \mathbf{v}_y^{(i)}=\mathbf{v}_y$ 
		\State \textbf{until} $\big |\frac{C(\mathbf{q}_1,\mathbf{q}_2,\mathbf{v}_y^{(i)})-C(\mathbf{q}_1,\mathbf{q}_2,\mathbf{v}_y^{(i-1)})}{C(\mathbf{q}_1,\mathbf{q}_2,\mathbf{v}_y^{(i)})}\big|<  {\epsilon_3} $
		\State \textbf{return} \{$\mathbf{v}_y, \mathbf{P}_{\mathrm{ com,1}}, \mathbf{P}_{\mathrm{ com,2}}$\}
	\end{algorithmic}
\end{algorithm}
\fi
The proposed procedure to solve sub-problem $\mathrm{(P.1.c)}$ is summarized in \textbf{Algorithm} \ref{alg:sca}, where the convex
approximation $\mathrm{ (\widetilde{P.1.c})}$ is solved using the Python convex optimization library CVXPY \cite{cvxpy}. It can be shown that \textbf{Algorithm} \ref{alg:sca}
converges to a local optimum of sub-problem $\mathrm{(P.1.c)}$ in polynomial computational time complexity \cite{sca_complexity}. In particular, \textbf{Algorithm} \ref{alg:sca} {involves} $3N-1$ optimization variables and, therefore, has a computational time complexity of $\mathcal{O}(M_3(3N-1)^{3.5})$, where $M_3$ is the number of iterations needed for convergence \cite{ipm_complexity}.
\subsection{Solution {of} Problem $\rm (P.1)$}
To solve problem $\rm (P.1)$, we use \ac{ao} by solving sub-problems $\rm (P.1.a)$,  $\rm (P.1.b)$, and  $\rm (P.1.c)$, alternately, see Figure \ref{fig:block_diagram}. To guarantee the convergence of the proposed {\ac{ao} \textbf{Algorithm} \ref{alg:ao}}, \textbf{Algorithm} \ref{alg:pso} {has to} be adjusted in each iteration $m$ to ensure {a} non-decreasing objective function. This can be achieved by adding one more particle to the initial \ac{pso} population of size $D$. The added particle represents the position of the slave \ac{uav} in the across-track {plane} position in the previous \ac{ao} iteration, i.e., $\mathbf{q}_2^{(m-1)}$, and the resulting \ac{pso} population is now of size $D+1$. {\newline In general, for} non-convex optimization problems, characterized by multiple local optima, \ac{ao} methods are vulnerable to {getting trapped in} low-quality solutions \cite{notes_ao}, which is the case for  problem  $\mathrm{(P.1)}$ due to {the} coupled optimization variables. In particular,  the objective function  is increasing \ac{wrt} the \ac{uav} velocity $\mathbf{v}_y$. {However}, increasing $\mathbf{v}_y$ results in constraining the \ac{uav} formation $\{\mathbf{q}_1,\mathbf{q}_2\}$ due to the \ac{snr} decorrelation {(see constraint $\mathrm{C6}$)}, yielding multiple local optima depending on the magnitude of the \ac{ao}  variable updates in each iteration. To avoid premature convergence to a low-quality solution, we propose to limit the {\ac{uav}} velocity update in iteration $m$ of the \ac{ao} algorithm by using a step size, denoted by ${\psi}$, as follows: 
\ifonecolumn
\begin{equation}\label{eq:step_size}
	\mathbf{v}^{(m)}_y=\begin{dcases}
		\mathbf{v}_y^{(m-1)}+{\psi}(\mathbf{v}_y-\mathbf{v}_y^{(m-1)}),\hspace{2mm}&\text{ if  feasible for $\mathrm{(P.1.c)}$, } \\
		\mathbf{v}_y, \hspace{2mm}&\text{ otherwise, $\forall m.$ }
\end{dcases}	
\end{equation}
\else 
	\begin{equation}\label{eq:step_size}
	\mathbf{v}^{(m)}_y=\begin{dcases}
		\mathbf{v}_y^{(m-1)}+{\psi}(\mathbf{v}_y-\mathbf{v}_y^{(m-1)}), &\text{if  feasible for $\mathrm{(P.1.c)}$, }\\
		\mathbf{v}_y, &\text{ otherwise, $\forall m.$ }
	\end{dcases}	
\end{equation}
\fi 
where $\mathbf{v}^{(m)}_y$ is the \ac{uav} velocity update in the $m^{\rm th}$ iteration of the \ac{ao} algorithm, and $\{\mathbf{v}_y,\mathbf{P}_{\rm com,1},\mathbf{P}_{\rm com,2}\}$ is the solution {of} $\mathrm{(P.1.c)}$ generated by \textbf{Algorithm} \ref{alg:sca}. Note that small ${\psi}$ values result in smaller updates for the velocity vector, whereas ${\psi}=1$ results in the classical \ac{ao} updates \cite{notes_ao}.\par
In \textbf{Algorithm} \ref{alg:ao}, we summarize all the steps of the solution of problem $\mathrm{(P.1)}$, see Figure \ref{fig:block_diagram}.  {Moreover, the convergence of \textbf{Algorithm} \ref{alg:ao} is proven in \textbf{Proposition} \ref{prop:convergence}.}
{
\begin{proposition}\label{prop:convergence}
	Algorithm \ref{alg:ao} is guaranteed to converge to a sub-optimal solution of problem $\mathrm{(P.1)}$. 
\end{proposition}
\begin{proof}
	Please refer to Appendix \ref{app:convergence}.
\end{proof}}
\ifonecolumn
\begin{algorithm}[h]
	\caption{$\psi$-Step Alternating Optimization Algorithm}\label{alg:ao}
	\begin{algorithmic}[1] 
		\State Set initial formation $\{\mathbf{q}_1^{(1)},\mathbf{q}_2^{(1)}\}$ and initial  resources $\{\mathbf{v}^{(1)}_y,\mathbf{P}^{(1)}_{\mathrm{com,1}},\mathbf{P}^{(1)}_{\mathrm{com,2}}\}$, iteration index $m=1$, and error tolerance $0< {\epsilon_4}  \ll 1$.
		\State \textbf{repeat}
		\State\hspace{\algorithmicindent}$ $ Set $ m=m+1$
		\State\hspace{\algorithmicindent}$ $ Determine coverage $ C_N(\mathbf{q}_1^{(m-1)},\mathbf{q}_2,\mathbf{v}_y^{(m-1)})$ and { set $\mathbf{q}_2^{(m)}=\mathbf{q}_2$}  by solving $\rm (P.1.a)$ for fixed \Statex\hspace{\algorithmicindent} $\{\mathbf{v}_y^{(m-1)},\mathbf{P}^{(m-1)}_{\mathrm{com,1}},\mathbf{P}^{(m-1)}_{\mathrm{com,2}},\mathbf{q}_1^{(m-1)}\}$  using \textbf{Algorithm} \ref{alg:pso}   
		\State\hspace{\algorithmicindent}$ $ Determine coverage $ C_N(\mathbf{q}_1,\mathbf{q}_2^{(m)},\mathbf{v}_y^{(m-1)})$ and { set $\mathbf{q}_1^{(m)}=\mathbf{q}_1$} by solving $\rm (P.1.b)$ for fixed \Statex\hspace{\algorithmicindent} $\{\mathbf{v}_y^{(m-1)},\mathbf{P}^{(m-1)}_{\mathrm{com,1}},\mathbf{P}^{(m-1)}_{\mathrm{com,2}},\mathbf{q}_2^{(m)}\}$  using \textbf{Algorithm} \ref{alg:polyblock}
		\State\hspace{\algorithmicindent}$ $ Determine coverage $ C_N(\mathbf{q}_1^{(m)},\mathbf{q}_2^{(m)},\mathbf{v}_y)$ and $\{\mathbf{v}_y, \mathbf{P}_{\mathrm{com,1}},\mathbf{P}_{\mathrm{com,1}}\}$ by solving $\rm (P.1.c)$ for fixed  \Statex\hspace{\algorithmicindent} formation $\{\mathbf{q}_1^{(m)},\mathbf{q}_2^{(m)}\}$  using \textbf{Algorithm} \ref{alg:sca}
		\State\hspace{\algorithmicindent} \textbf{if}  $\{\mathbf{v}_y^{(m-1)}+{\psi}(\mathbf{v}_y-\mathbf{v}_y^{(m-1)}),\mathbf{P}_{\rm com,1},\mathbf{P}_{\rm com,2}\}$  is  feasible for $\mathrm{(P.1.c)}$ \textbf{then}	\State\hspace{\algorithmicindent}\hspace{\algorithmicindent}$ $ Set $\mathbf{v}_y^{(m)}=\mathbf{v}_y^{(m-1)}+{\psi}(\mathbf{v}_y-\mathbf{v}_y^{(m-1)}),$ $\mathbf{P}^{(m)}_{\mathrm{com,1}} =\mathbf{P}_{\mathrm{com,1}},$ and $\mathbf{P}^{(m)}_{\mathrm{com,2}} =\mathbf{P}_{\mathrm{com,2}}$
		\State\hspace{\algorithmicindent} \textbf{else}
		\State\hspace{\algorithmicindent}\hspace{\algorithmicindent}$ $ Set $\mathbf{v}_y^{(m)}=\mathbf{v}_y,$ $\mathbf{P}^{(m)}_{\mathrm{com,1}} =\mathbf{P}_{\mathrm{com,1}},$ and $\mathbf{P}^{(m)}_{\mathrm{com,2}} =\mathbf{P}_{\mathrm{com,2}}$
		\State \textbf{until} $\big |\frac{ C_N(\mathbf{q}_1^{(m)},\mathbf{q}_2^{(m)},\mathbf{v}_y^{(m)})-C_N(\mathbf{q}_1^{(m-1)},\mathbf{q}_2^{(m-1)},\mathbf{v}_y^{(m-1)})}{C_N(\mathbf{q}_1^{(m)},\mathbf{q}_2^{(m)},\mathbf{v}_y^{(m)})}\big|\leq {\epsilon_4} $
		\State \textbf{return} solution $\{\mathbf{q}_1^{(m)}, \mathbf{q}_2^{(m)},\mathbf{P}^{(m)}_{\mathrm{com,1}},\mathbf{P}^{(m)}_{\mathrm{com},2}, \mathbf{v}_y^{(m)}\}$
	\end{algorithmic}
\end{algorithm} 
\fi
{Leveraging the} computational complexity analysis of the different sub-problems, \textbf{Algorithm} \ref{alg:ao} has $\mathcal{O}\Big(M\big(M_1 D N + M_2 {\log_2}(1/{\epsilon_2} ) + M_3 (3N-1)^{3.5} \big)\Big)$ worst-case computational time complexity, where $M$ is the number of iterations required for \textbf{Algorithm} \ref{alg:ao} to converge. {\textbf{Algorithm} \ref{alg:ao} has a polynomial computational complexity, and thus, is considered to be a low-complexity algorithm. In practice, convergence is achieved in just a few iterations.} Consequently, due to the low complexity of \textbf{Algorithm} 5, the { $\epsilon_5$-optimal step size ${\psi}$ can be determined by performing an exhaustive search in $[0,1]$ with precision $\epsilon_5$.} Note that \ac{pso} \textbf{Algorithm} \ref{alg:pso}, being a stochastic approach, {can achieve the same objective function {value} for sub-problem $\mathrm{(P.1.a)}$ {for} different solutions. {Therefore,} \textbf{Algorithm} \ref{alg:ao} might lead to slightly different results for the overall problem for $\mathrm{(P.1)}$ given the same input, which is a known property of \ac{ao} {algorithms} when involving a stochastic algorithm\cite{notes_ao}.}
\ifonecolumn
\begin{table}
	\caption{System parameters \cite{propulsion,victor,amine2,snr_equation,coherence1}.}
	\label{tab:system_parameters}
	\begin{adjustbox}{width=\textwidth}
		\begin{tabular}{?cc?cc?}
			\Xhline{3\arrayrulewidth}
			\multicolumn{1}{?c|}{Parameter}                     & Value                    & \multicolumn{1}{c|}{Parameter} & Value \\ \Xhline{3\arrayrulewidth}
			\multicolumn{2}{?c?}{\cellcolor[HTML]{EFEFEF}\ac{uav} parameters} & \multicolumn{2}{c?}{\cellcolor[HTML]{EFEFEF}Radar parameters}    \\  \hline
			\multicolumn{1}{?c|}{Total number of time slots, $N$ }                             & 80     & \multicolumn{1}{c|}{  Reference target $x$-coordinate, $x_t$       }         & 20 m          \\ \hline
			\multicolumn{1}{?c|}{   Maximum and minimum \ac{uav}  altitudes, $z_{\mathrm{max}}$ and   $z_{\mathrm{min}}$    }     & 100 and 1 m     & \multicolumn{1}{c|}{ Time slot and pulse duration, $\delta_t$ and $\tau_p$  }         & 1 and $10^{-6}$ s    \\ \hline
			\multicolumn{1}{?c|}{Maximum and minimum \ac{uav}  velocities, $v_{\mathrm{max}}$ and $v_{\mathrm{min}}$     }     & 10   and   0.1  m/s       & \multicolumn{1}{c|}{Master look angle and radar beamwidth, $\theta_1$ and $\Theta_{3\mathrm{dB}}$           }         & 45° and 30°       \\ \hline
			\multicolumn{1}{?c|}{ UAV battery capacity, $E_{\mathrm{max}}$        }     & 122.2 Wh             & \multicolumn{1}{c|}{ Minimum and maximum slave look angle, $\theta_{\rm min}$ and $\theta_{\rm max}$     }         & 15° and 75°		     \\ \hline
			\multicolumn{1}{?c|}{  Minimum baseline distance, $b_{\mathrm{min}}$      }     & 2 m            & \multicolumn{1}{c|}{Radar transmit and receive gains, $G_t$ and $G_r$      }         & 5 dBi  \\ \hline
			\multicolumn{1}{?c|}{   UAV weight in Newton, $W_u$     }     & 60 N             & \multicolumn{1}{c|}{  Radar transmit power  and noise figure, $P_{t,1}=P_{t,2}$ and  $F$    }         &  10 dBm   and 7 dB       \\ \hline
			\multicolumn{1}{?c|}{  Air density, $\rho$      }     & 1.225 kg/m$^3$           & \multicolumn{1}{c|}{    System, azimuth, and atmospheric losses, $L_{\mathrm{sys}}$, $L_{\mathrm{azm}}$, and $L_{\mathrm{atm}}$         }         & 2, 2, and 0 dB     \\ \hline
			\multicolumn{1}{?c|}{   Fuselage drag ratio, $d_0$    }     & 0.6       & \multicolumn{1}{c|}{      Pulse repetition frequency, $\mathrm{PRF}$       }         & 100 Hz  \\ \hline
			\multicolumn{1}{?c|}{  Profile drag coefficient, $\delta_u$     }     & 0.0012       & \multicolumn{1}{c|}{   Radar wavelength, $\lambda$       }         &  0.12 m    \\ \hline
			\multicolumn{1}{?c|}{   Rotor radius, $R$   }     &0.4 m       & \multicolumn{1}{c|}{   Radar bandwidth and center frequency, $B_{\mathrm{Rg}}$ and $f_0$       }         &   3 and 2.5 GHz    \\ \hline
			\multicolumn{1}{?c|}{   Rotor disc area, $A_e$    }     & 0.503 m$^2$       & \multicolumn{1}{c|}{      Number of independent looks, $n_L$          }         &    4 \\ \hline
			\multicolumn{1}{?c|}{  Rotor solidity ratio, $s$     }     & 0.05        & \multicolumn{1}{c|}{   Radar system temperature, $T_{\mathrm{sys}}$       }         &400 K      \\ \hline
			\multicolumn{1}{?c|}{   Blade angular velocity, $\Omega$    }     & 300 rad/s       & \multicolumn{1}{c|}{   Normalized backscatter coefficient, $\sigma_0$       }         & -10 dB     \\ \hline
			\multicolumn{1}{?c|}{    Tip speed of the rotor blade, $U_{\mathrm{tip}}$   }     &120 m/s       & \multicolumn{1}{c|}{    Minimum \ac{snr} and baseline decorrelation,$\gamma_{\mathrm{SNR}}^{\mathrm{min}}$ and $\gamma_{\mathrm{Rg}}^{\mathrm{min}}$    }         &  0.8   \\ \hline
			\multicolumn{1}{?c|}{  Incremental correction factor to induced power, $k_u$     }     & 0.1       & \multicolumn{1}{c|}{ Other decorrelation sources, $\gamma_{\rm other}$      }         &  0.9   \\\hline \multicolumn{2}{?c?}{\cellcolor[HTML]{EFEFEF}Communication parameters} & \multicolumn{1}{c|}{   Minimum \ac{hoa} and maximum height error, $h_{\mathrm{amb}}^{\mathrm{min}}$ and $\Delta_h^{\rm max}$     }         & 1 and 0.11 m     \\ \hline
			\multicolumn{1}{?c|}{\ac{gs} 3D location, $\mathbf{g}=(g_x,g_y,g_z)^T$     }     &  $(-100,-270,5)^T$ m            & \multicolumn{2}{c?}{\cellcolor[HTML]{EFEFEF}Algorithm parameters}  \\ \hline
			\multicolumn{1}{?c|}{ Master communication bandwidth, $B_{c,1}$     }     &1 GHz         & \multicolumn{1}{c|}{    \ac{pso} population and maximum number of iterations, $D$ and $M_1$      }         &    $2\times10^3$ and $1000$      \\ \hline
			\multicolumn{1}{?c|}{  Slave communication bandwidth, $B_{c,2}$   }     &  1 GHz      & \multicolumn{1}{c|}{   \ac{pso} cognitive and social constants, $c_1$ and $c_2$     }         &   0.1 and 0.2         \\ \hline
			\multicolumn{1}{?c|}{  Communication reference channel gains, ${\beta_{c,1}}={\beta_{c,2}}$  }     &18.751 dB   & \multicolumn{1}{c|}{     \ac{pso} maximum particle velocity, $v_{\rm PSO}^{\rm max}$       }         &  20 m/s      \\ \hline
			\multicolumn{1}{?c|}{  Number of bits per complex sample, $n_B$  }     & 4  & \multicolumn{1}{c|}{ Constant $O$ and number of realizations $\mathcal{N}$ }         &  500  and $10^3$          \\ \hline
			\multicolumn{1}{?c|}{    Maximum communication transmit power, $P_{\mathrm{com}}^{\mathrm{max}}$       }     &   40 dBm         & \multicolumn{1}{c|}{      Error tolerance, ${\epsilon_1=\epsilon_2=\epsilon_3=\epsilon_4} $ and $\epsilon_5$    }         & $10^{-4}$ and  $10^{-2}$  \\ \Xhline{3\arrayrulewidth}
		\end{tabular}
	\end{adjustbox}
\end{table}
\fi 
\section{Simulation Results and Discussion }
In this section, we present simulation results for the proposed \ac{uav} formation and resource allocation optimization algorithm. Unless specified otherwise, the system parameters specified in Table \ref{tab:system_parameters} {are employed}. As explained in the previous section, since the {algorithm proposed} to solve $\mathrm{(P.1)}$ is based on stochastic optimization, different solutions {may} be {obtained} for the same input. Therefore, simulation results are averaged over $10^3$ realizations. To assess the performance of the proposed {algorithm}, we {consider} the following benchmark schemes: 
\begin{itemize}
	\item \textbf{Benchmark scheme 1} ({\it {Classical} \ac{ao}}): For this scheme, we employ the classical \ac{ao} algorithm \cite{notes_ao} to solve problem $\mathrm{(P.1)}$, i.e., without adjusting step size ${\psi}$. {This} is equivalent to using \textbf{Algorithm} \ref{alg:ao} {with} ${\psi}=1$.\begin{comment}
	\item \textbf{Benchmark scheme 2} ({\it {Optimistic}-master-swath}): Here, we  find the master \ac{uav} location $\mathbf{q}_1^{\rm bench2}$ achieving the largest master swath width and satisfying constraints  $\mathrm{C1-C15}$, see Appendix \ref{app:optimistic-master-swath}. Then, for given  $\mathbf{q}_1=\mathbf{q}_1^{\rm bench2}$, we alternately optimize the \ac{uav} resources and the across-track position  of the slave platform, i.e., $\mathrm{(P.1.a)}$ and $\mathrm{(P.1.c)}$, respectively,  until convergence. {Here, the value of ${\psi}$ leading to the best \ac{insar} coverage is searched via exhaustive search in $[0,1]$} { with precision $\epsilon_5$}.\end{comment}
	\item \textbf{Benchmark scheme 2} ({\it {Fixed} steady speed}): As in \cite{conf2}, in this scheme, we  {employ}  a fixed and steady radar platform velocity, i.e.,  $\mathbf{v}_y$ is not optimized {but} is set to $\mathbf{v}_y[n]=4$ m/s$, \forall n$. We use \ac{ao} to optimize the communication powers and the \ac{uav} formation. {This} is equivalent to using \textbf{Algorithm} \ref{alg:ao} with ${\psi}=0$.
	\item \textbf{Benchmark scheme 3} ({\it {Fixed} slave look angle}): Similar to \cite{conf2}, we {adopt} a predefined  look angle for the slave platform, $\theta_2(\mathbf{q}_2)=\frac{\pi}{4}$. {Problem} ($\mathrm{P.1}$) is solved with a corresponding modified version of \textbf{Algorithm} \ref{alg:ao}. {Here, the value of $\psi$ leading to the best \ac{insar} coverage is found via exhaustive search in $[0,1]$  with precision $\epsilon_5$.}\par
\end{itemize}

\begin{figure}
	\centering
	\ifonecolumn
	\includegraphics[width=4.5 in]{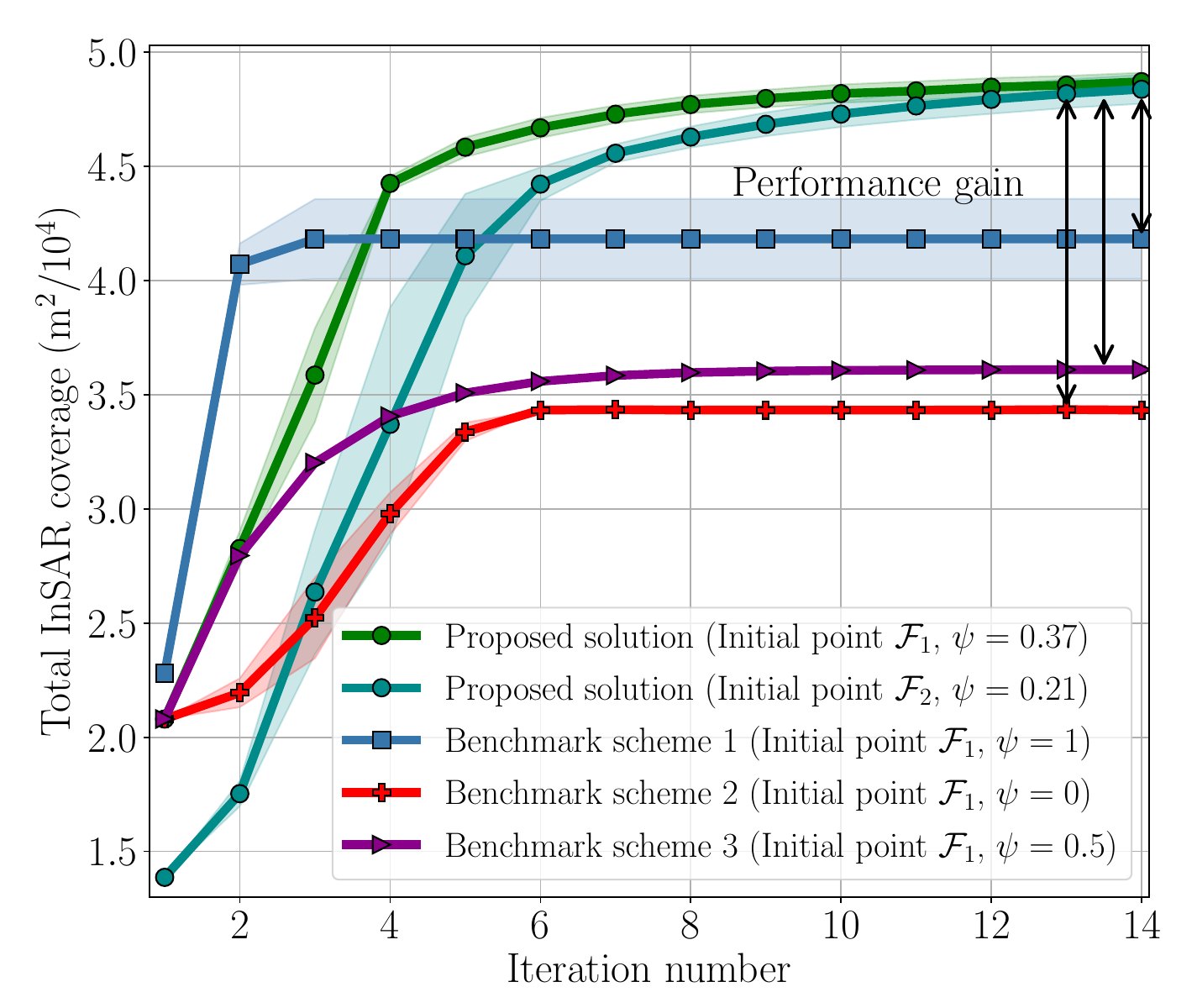}
	\else 
	\includegraphics[width=0.85\columnwidth]{figures//InSAR_coverage.pdf}
	\fi
	\caption{Interferometric sensing coverage achieved by the proposed solution and {benchmark} schemes versus iteration number for different initial points, denoted by  $\mathcal{F}_1$  and $\mathcal{F}_2$, and given by $\mathcal{F}_1=\Big\{\mathbf{q}^{(1)}_1=(-40,60)^T$ m, $\mathbf{q}^{(1)}_2=(-45,50)^T$ m, $P^{(1)}_{\rm com,1}[n]=P^{(1)}_{\rm com,2}[n]=37.78$ dBm, $v^{(1)}_y[n]=3.8$ m/s, $\forall n\Big\}$ and  $\mathcal{F}_2=\Big\{\mathbf{q}^{(1)}_1=(-20,40)^T$ m, $\mathbf{q}^{(1)}_2=(-45,50)^T$ m, $P^{(1)}_{\rm com,1}[n]=P^{(1)}_{\rm com,2}[n]=37.78$ dBm, $v^{(1)}_y[n]=3.8$ m/s, $\forall n\Big\}$. Solid lines represent mean values while the shaded area corresponds to one standard deviation, i.e., $1\sigma$. {For the proposed algorithm, the ${\psi}$ values are optimized based on an exhaustive search.}}
	\label{fig:convergence}
\end{figure}
\par In Figure \ref{fig:convergence}, we investigate the convergence of the proposed \ac{ao} algorithm {and} the different benchmark schemes by plotting the total \ac{insar} coverage as a function of the number of iterations. The figure shows that the classical \ac{ao} algorithm, i.e., benchmark scheme 1, outperforms benchmark schemes 2 and 3 due to {the joint optimization of} the \ac{uav} formation and resource allocation. In fact, benchmark scheme {2 exhibits} the lowest performance, underlining the importance of optimizing the instantaneous velocity of both \acp{uav}  and the slave look angle to achieve maximal sensing coverage, respectively. Furthermore, Figure \ref{fig:convergence} shows that{, after convergence,} the proposed solution outperforms all benchmark schemes achieving  performance {gains of  14.8 \%, 41.21 \%, and 35.74 \%,} compared to benchmark schemes 1, 2, and 3, respectively. In particular, the superiority of the proposed solution compared to benchmark scheme 1 is attributed to {the} tuning {of} the step size ${\psi}$ to avoid premature convergence to inferior sub-optimal solutions. This performance gain comes at the cost of {a lower} convergence speed, as the proposed solution requires more iterations to converge. {Nevertheless,} overall, all benchmark schemes as well as the proposed solution {converge after} a {relatively} small number of iterations. Lastly, for different initial points $\mathcal{F}_1$ and $\mathcal{F}_2$, the proposed \ac{ao}-based algorithm achieves similar performance {confirming} resilience to the initialization.\par 

\begin{figure}[h]
	\centering
	\ifonecolumn
	\includegraphics[width=4 in]{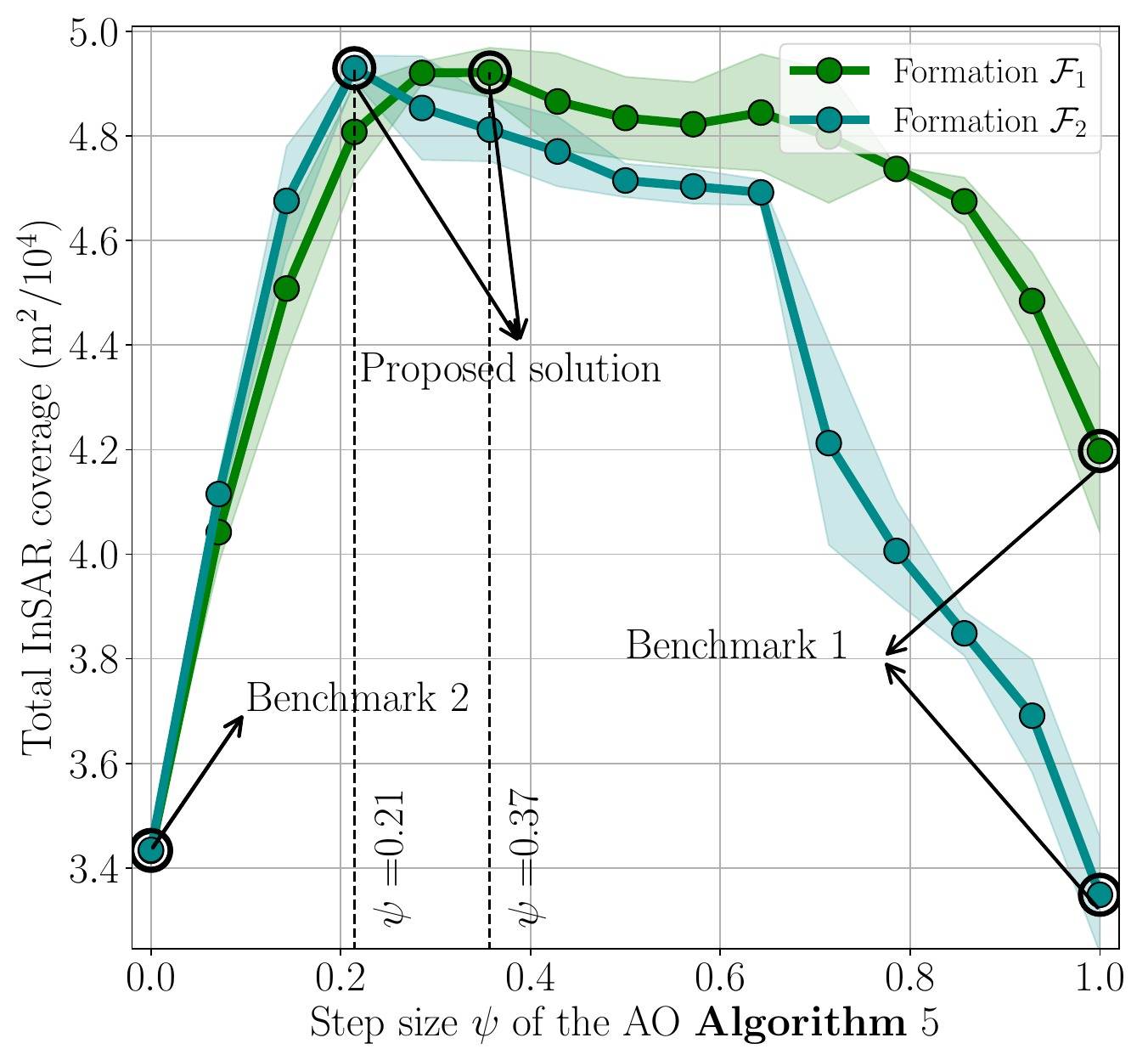}
	\else 
	\includegraphics[width=0.87\columnwidth]{figures/BCD_step_size2.pdf}
	\fi
	\caption{Interferometric sensing coverage versus  step size ${\psi}$ of \textbf{Algorithm} \ref{alg:ao} for different initial points, denoted by $\mathcal{F}_1$ and $\mathcal{F}_2$, see Figure \ref{fig:convergence}. Solid lines represent mean values while the shaded area corresponds to one standard deviation, i.e., $1\sigma$.}
	\label{fig:step_size}
\end{figure}
In Figure \ref{fig:step_size}, we {show} the total \ac{insar} coverage as a function of step size ${\psi}$ defined in (\ref{eq:step_size}). The figure confirms that the performance of the {conventional} \ac{ao} {algorithm}, i.e., {for} ${\psi}=1$, can be significantly improved by tuning ${\psi}$ in $[0,1]$. Moreover, the optimal step size depends on the system parameters and the \ac{ao} initialization process. For instance, the best step size for initial point $\mathcal{F}_1$ is  ${\psi}={0.37}$, while for $\mathcal{F}_2$, {it is}  ${\psi}={0.21}$.  Figure \ref{fig:step_size} {reveals} that  values of ${\psi}$  close to 1 and 0 {are not preferable}. On the one hand, a value of ${\psi}$ close to 1 results in {a fast increase of the} velocity in {the first few} iterations to improve coverage. {However, this} leads to low \acp{snr} for the master and slave \acp{uav}, i.e., a high \ac{snr} decorrelation, see (\ref{eq:snr_decorrelation}). Therefore, the \ac{uav} formation in the across-track cannot be {properly} optimized and both drones {have to be positioned} close to the target area  to keep short master and slave slant ranges that compensate for the impact of the high velocity on the \ac{snr}. On the other hand, a value of ${\psi}$ close to 0 results in {slow} velocity updates, which allows for optimizing the \ac{uav} formation to improve coverage by opting for large master and slave slant ranges. However, the velocity can {subsequently not be} increased due to the \ac{snr} decorrelation constraint $\mathrm{C6}$, resulting in poor performance.  The proposed {approach} overcomes this issue by adequately selecting ${\psi}$ so that the updates of {the different} optimization variables are well-balanced in each iteration  to achieve maximum sensing coverage.\par
\begin{figure}[h]
	\centering
	\ifonecolumn
	\includegraphics[width=4 in]{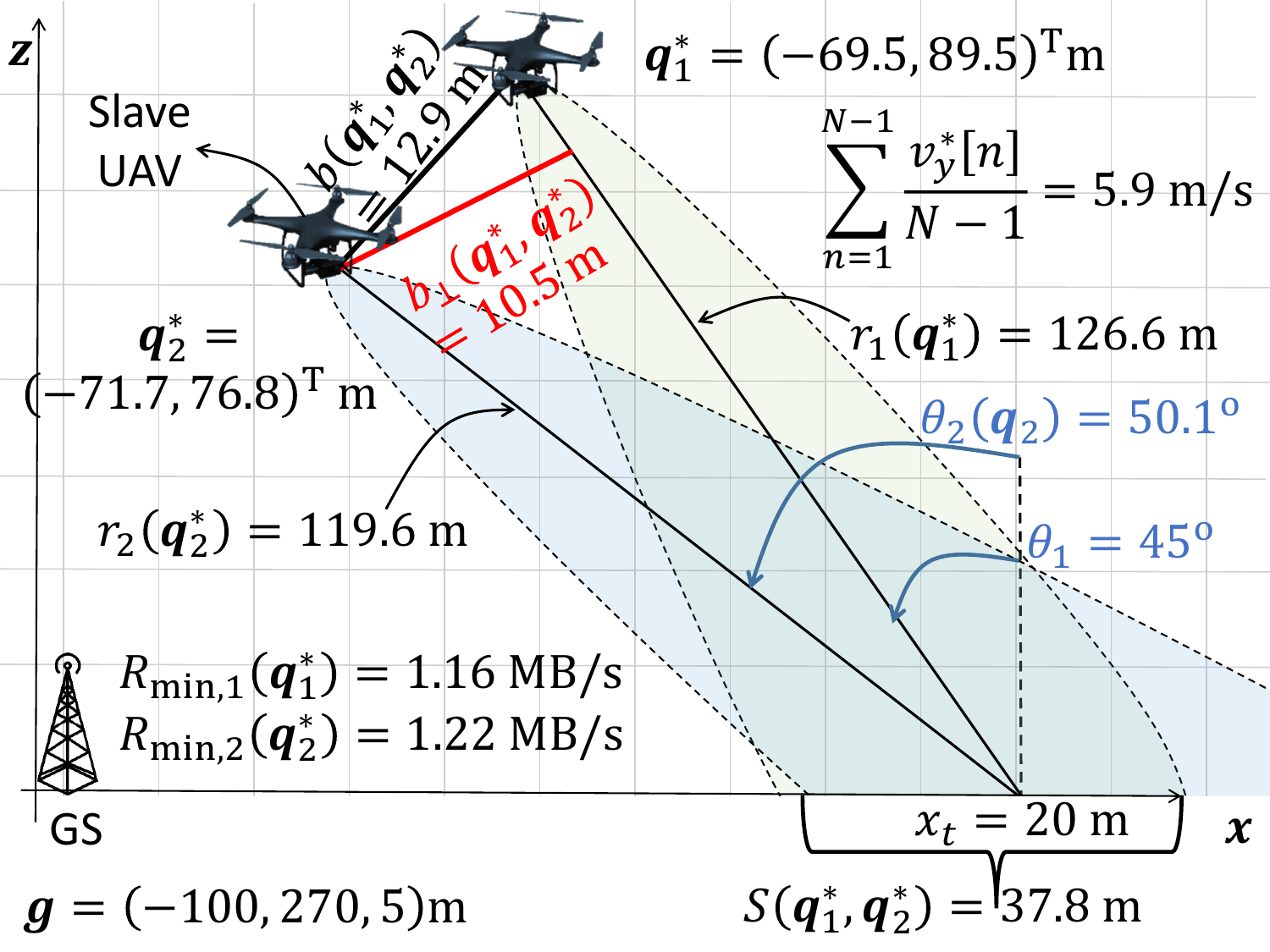}
	\else
	\includegraphics[width=0.80\columnwidth]{figures/Visual.pdf}
	\fi
	\caption{Illustration of optimal \ac{uav} formation in across-track plane with key parameter values (not to scale) specified.}
	\label{fig:visual}
\end{figure}
{In Figure \ref{fig:visual}, we illustrate a \ac{uav} formation in the across-track plane obtained with \textbf{Algorithm} \ref{alg:ao}. Note that \textbf{Algorithm} 5 may find different solutions leading to the same performance, therefore, the solution presented in Figure \ref{fig:visual} is not unique. For the shown example, the perpendicular baseline is 10.5 m, leading to a relative height error of 10.5 cm and a total \ac{insar} coverage of $C_N=$4.95$\times 10^4$ m$^2$.}\par
\begin{figure}[h]
	\centering
	\ifonecolumn
	\includegraphics[width=4 in]{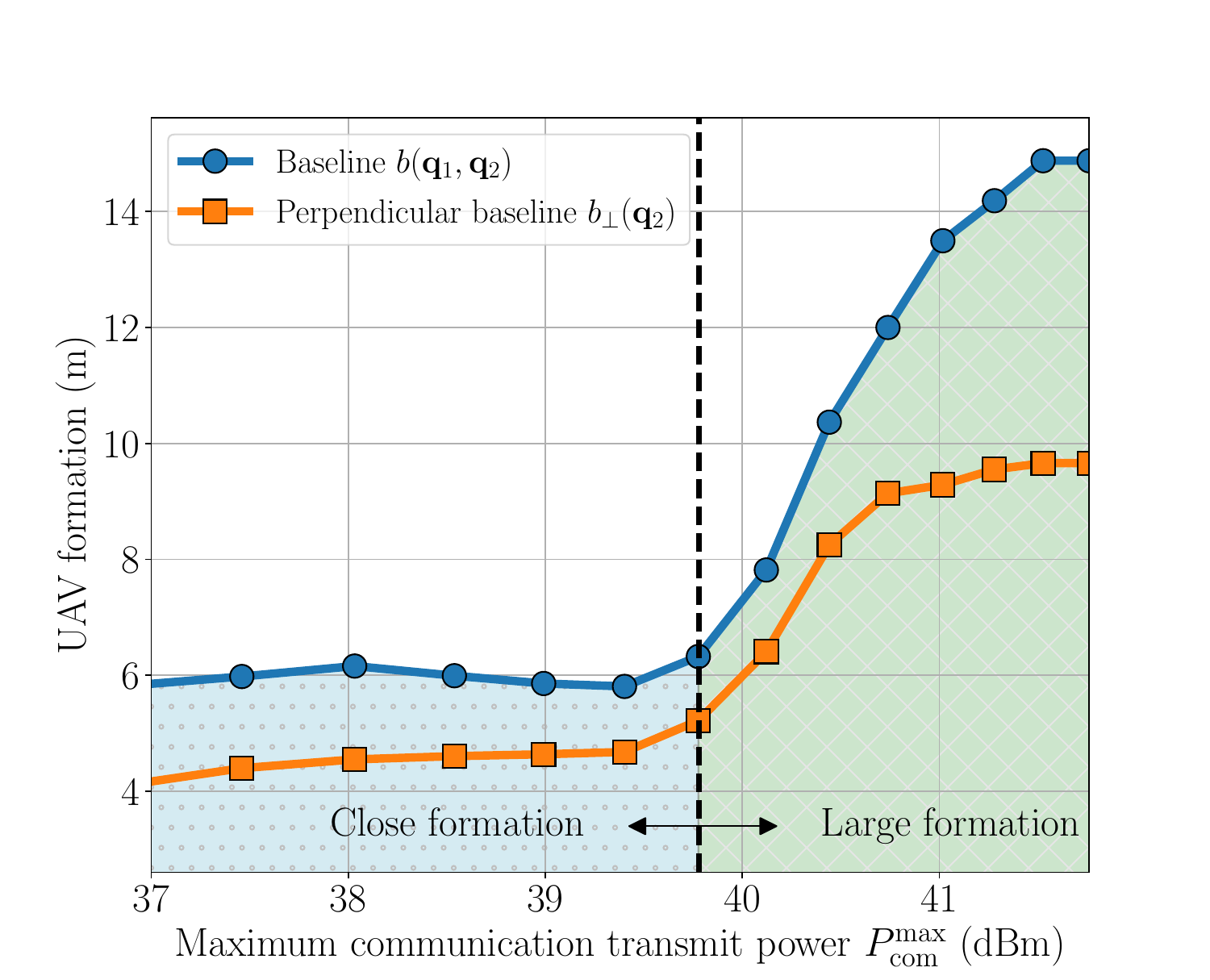}
	\else
	\includegraphics[width=0.85\columnwidth]{figures/Baseline_vs_Pcom_max.pdf}
	\fi
	\caption{Optimal \ac{uav} formation represented by the \ac{uav} baseline and perpendicular baseline versus the maximum communication transmit power for ${\beta_{c,1}}={\beta_{c,2}}=19.3$ dB. }
	\label{fig:baseline_vs_Pcom_max}
\end{figure}
In Figure \ref{fig:baseline_vs_Pcom_max}, we investigate the impact of the maximum communication transmit power $P_{\rm com}^{\rm max}$ on the optimal \ac{uav} bistatic formation {in terms of} the interferometric baseline $b$ and its perpendicular component $b_{\bot}$. The figure shows that the \ac{uav} formation employed for sensing is heavily {affected} by the communication transmit power used to offload the \ac{sar} data. In particular, for low maximum communication transmit {powers}, i.e., $P_{\rm com}^{\rm max}\leq {39.8}$ dBm  {(blue-colored region)}, the proposed optimization algorithm  {prefers} a close formation with short interferometric  {baselines as small as} {6} m. {In contrast}, high maximum communication transmission  {powers}, i.e., $P_{\rm com}^{\rm max}\geq {39.8}$ dBm {(}green-colored region{)}, result in a large \ac{uav} formation,  {with baselines as large as}  {14} m. This is because for low maximum communication transmit {powers}, {data }offloading is the performance bottleneck {and}, thus, \textbf{Algorithm} \ref{alg:ao} tends to {position} both \acp{uav} near the target area, resulting in less sensing data. This leads to a short master slant range $r_1$ and, therefore, a shorter perpendicular baseline is required to satisfy constraint $\mathrm{C8}$, see (\ref{eq:height_of_ambiguity}).   { On the other hand}, when the maximum communication transmit power is sufficiently high, the proposed solution optimizes  {the} \ac{uav} formation {by flying at higher altitudes} for better coverage, leading to large \ac{uav} {formations}. This results in higher sensing data rates for the slave and master \acp{uav}. Yet, all data from both drones can be successfully offloaded {because of the high} communication transmission power {available}.\par 
\begin{figure}[h]
	\centering
	\ifonecolumn
	\includegraphics[width=4 in]{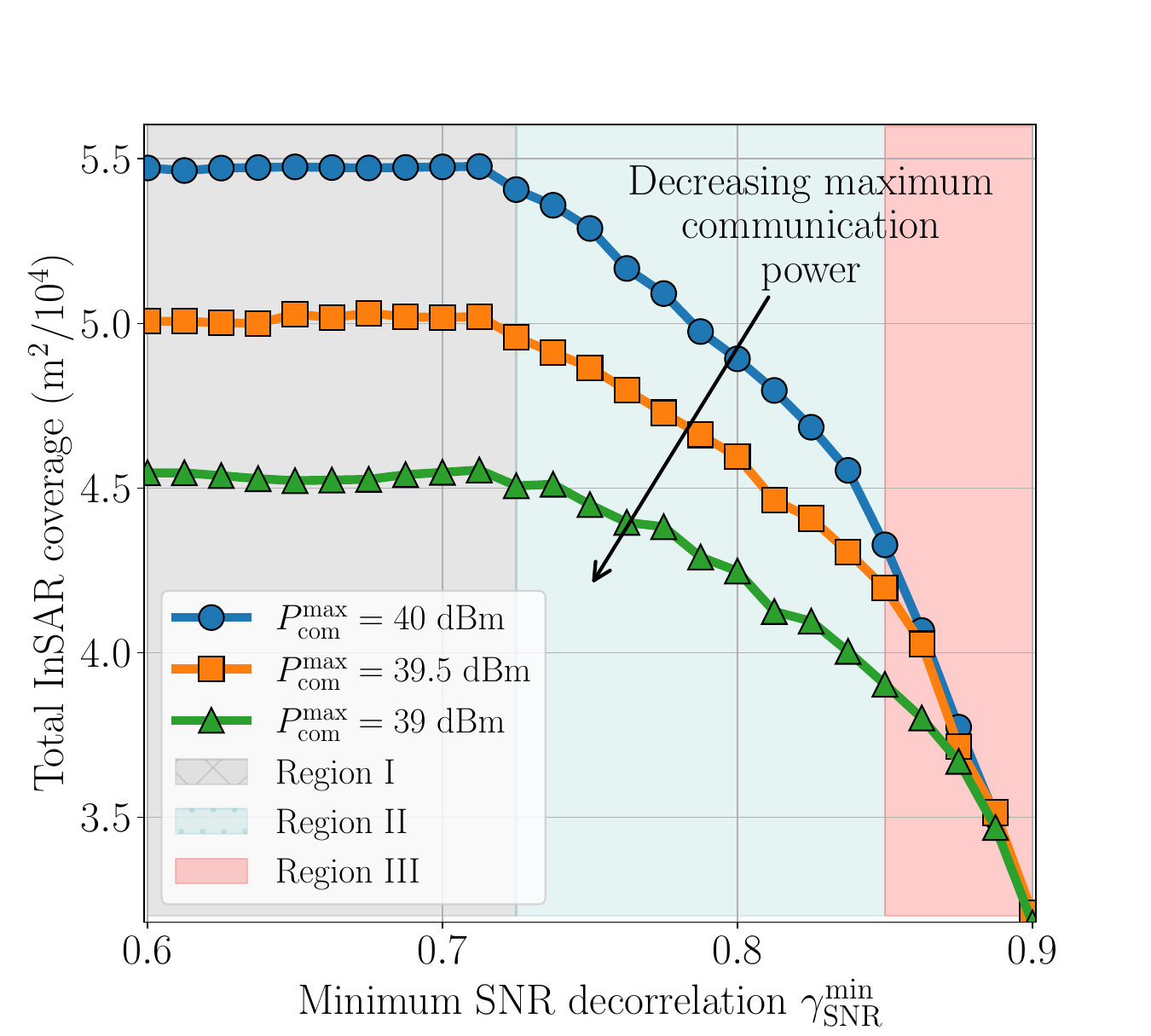}
	\else
	\includegraphics[width=0.9\columnwidth]{figures/Coverage_vs_SNR_decorrelation.pdf}
	\fi
	\caption{Interferometric sensing coverage versus the  {minimum required} \ac{snr} decorrelation $\gamma_{\rm SNR}^{\rm min}$ for different maximum communication transmit  {powers} $P_{\mathrm{com}}^{\mathrm{max}}$.}
	\label{fig:coverage_vs_snr_decorrelation}
\end{figure}
In Figure \ref{fig:coverage_vs_snr_decorrelation}, we  {show} the optimal \ac{insar} coverage  {obtained with \textbf{Algorithm} \ref{alg:ao}} as  {a}  function of the minimum  {required} \ac{snr} decorrelation for different maximum communication transmit  {powers}. The figure  {reveals} that the total sensing coverage decreases  {as the minimum required}  \ac{snr} decorrelation  {increases}, i.e.,  {as the} \ac{sar} coherence requirement  {becomes more stringent}. We distinguish between {three} regions in Figure \ref{fig:coverage_vs_snr_decorrelation}, namely Regions {I, II and III}. {In Region I, a high SNR decorrelation is allowed and  increasing the maximum communication transmit power results in better sensing coverage due to the extended  range  for communications. Therefore, in Region I, the system performance bottleneck is the air-to-ground communication. In Region II,  {stricter} sensing requirements in terms of \ac{sar} coherence  {are} imposed and the maximum communication power has a {more} moderate impact on the coverage compared to Region I. {Finally}, in Region III, a high coherence is required and constraint $\mathrm{C6}$ becomes active, i.e., sensing becomes the system performance bottleneck.  Consequently, the maximum communication transmit power has {little} impact on the coverage performance, which highlights the interplay between communication and sensing in the considered system.\par 
\begin{figure}[h]
	\centering
	\ifonecolumn
	\includegraphics[width=4 in]{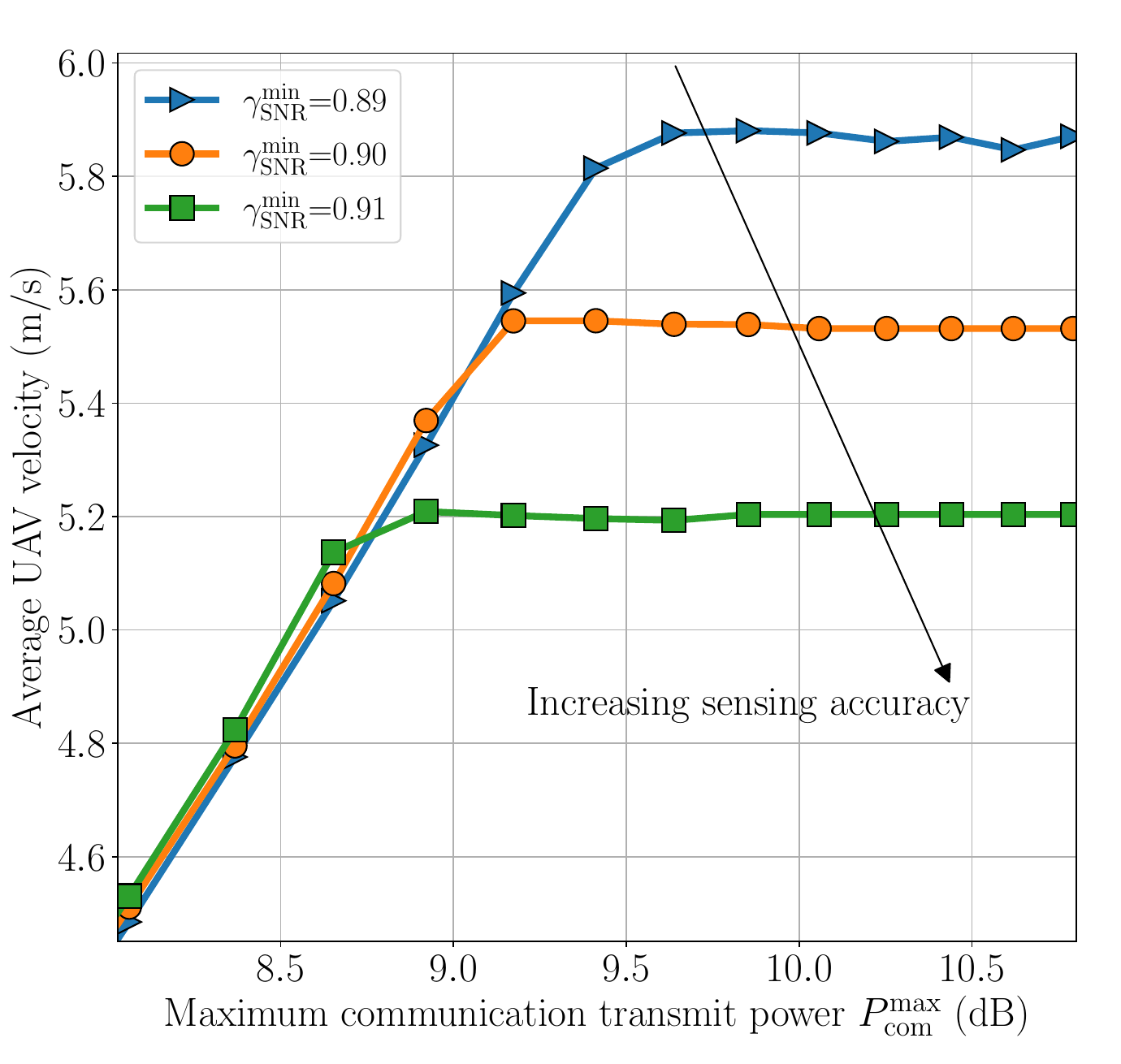}
	\else
	\includegraphics[width=0.85\columnwidth]{figures/Average_velocity_vs_Pcom_max.pdf}
	\fi
	\caption{Average \ac{uav} velocity as function of the maximum communication transmit power $P_{\rm com}^{\rm max}$ for different values of the minimum SNR decorrelation $\gamma_{\rm SNR}^{\rm min}$.}
	\label{fig:average_velocity_vs_maximum_communication_power}
\end{figure}
{In Figure \ref{fig:average_velocity_vs_maximum_communication_power}, we depict the average \ac{uav} velocity, given by ${\frac{1}{N-1}}\sum\limits_{n=1}^{N-1}v_y[n]$, as function of the maximum communication transmit power for different minimum SNR decorrelation thresholds. The figure shows that the average \ac{uav} velocity increases with the maximum communication transmit power due to the increased communication range for both \acp{uav}, which {enables} longer trajectories in azimuth direction. Moreover, the figure shows a strong impact of the sensing requirement on the drone speed. In fact, to achieve higher \ac{insar} coherence, {a lower} speed is preferred, as it allows for {improved} \ac{uav} \acp{snr}.  This further confirms the need for slowly moving drones for \ac{insar} imaging applications, especially when highly reliable sensing is demanded.\par}
\begin{figure}[h]
	\centering
	\ifonecolumn
	\includegraphics[width=3.5in]{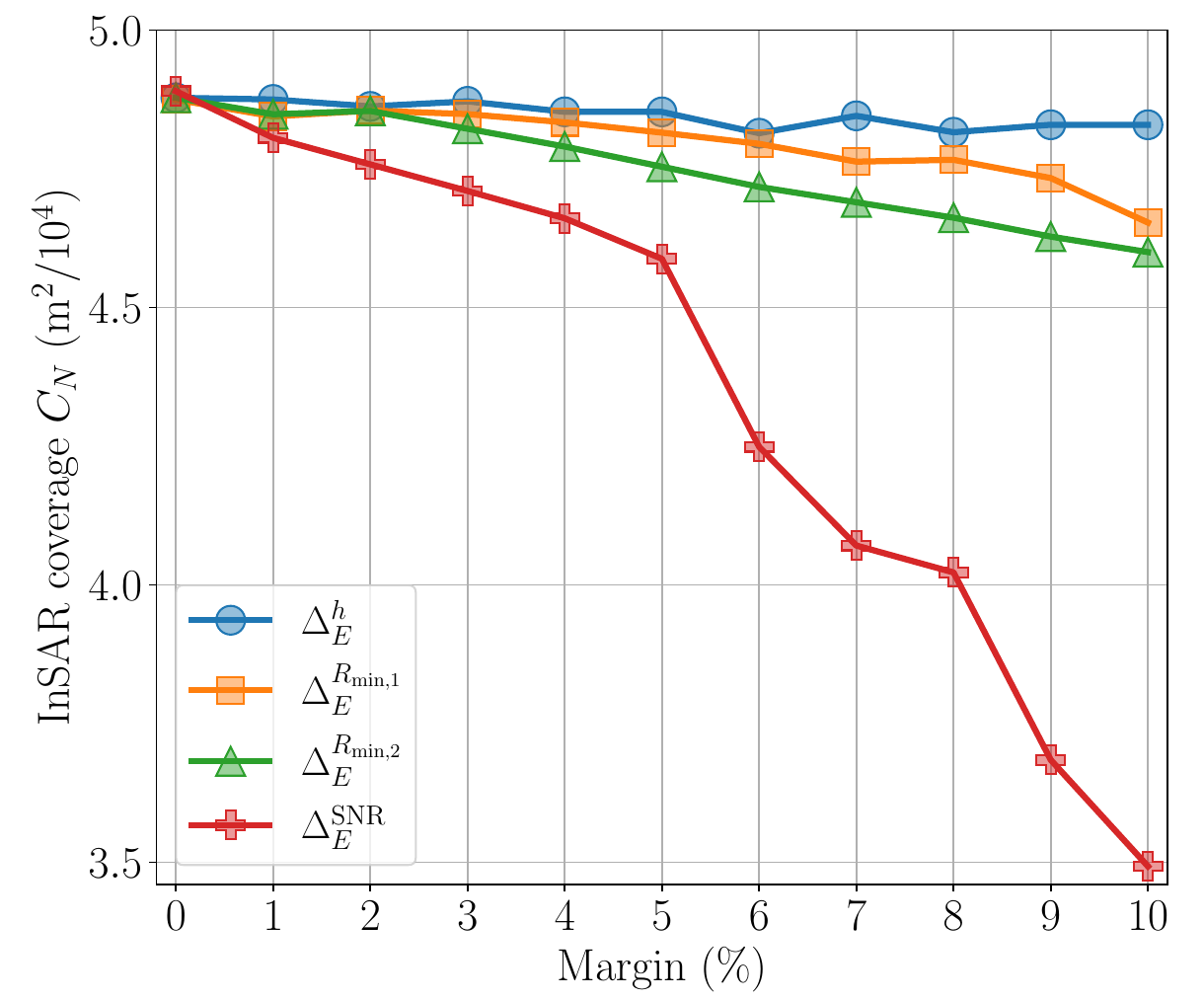}
	\else
	\includegraphics[width=0.8\columnwidth]{figures/sensitivity2.pdf}
	\fi
	\caption{Sensitivity analysis of \ac{insar} coverage when accounting for parameter uncertainties by making the constraints of problem $\mathrm{(P.1)}$ more stringent.}
	\label{fig:sensitivity2}
\end{figure}
{For the optimization of the considered \ac{uav}-\ac{insar} system, we had to assume specific values for the system parameters, see Table \ref{tab:system_parameters}. In practice, the actual values of these parameters may deviate from the assumed values. As a result, the solution obtained with \textbf{Algorithm} \ref{alg:ao} for the assumed values may violate the constraints of problem $\mathrm{(P.1)}$ for the actual parameter values. One approach to deal with this issue and make the proposed design robust to parameter uncertainties is to add safety margins to the constraints of problem  $\mathrm{(P.1)}$. For example, the minimum \ac{snr} decorrelation, $\gamma_{\mathrm{SNR}}^{\mathrm{min}}$, may be increased by a certain percentage value, $\Delta_E^{\mathrm{SNR}} \times 100 \%$, such that the new minimum required \ac{snr} becomes $(1+\Delta_E^{\mathrm{SNR}}) \gamma_{\rm SNR}^{\rm min}$. This ensures that after optimization based on the assumed parameter values, the $\mathrm{SNR}$ achieved for the actual parameter values does not fall below $\gamma_{\rm SNR}^{\rm min}$ as long as the deviation between assumed and actual values does not become exceedingly large. The price to be paid for this increased robustness is a decrease of the achievable coverage. In Figure \ref{fig:sensitivity2}, we investigate the impact of making constraints $\mathrm{C6}$, $\mathrm{C9}$, and $\mathrm{C11}$ more stringent to improve the robustness of the obtained solution. In particular, we increase $\gamma_{\rm SNR}^{\rm min}$, $R_{\mathrm{min},1}$, and $R_{\mathrm{min},2}$  to $(1+\Delta_E^{\mathrm{SNR}})\gamma_{\rm SNR}^{\rm min}$, $(1+\Delta_E^{R_{\mathrm{min},1}})R_{\mathrm{min},1}$, and $(1+\Delta_E^{R_{\mathrm{min},2}})R_{\mathrm{min},2}$, respectively. Additionally, we decrease $\Delta h^{\rm max}$ to $(1-\Delta_E^{h})\Delta h^{\rm max}$. Then, we show the \ac{insar} coverage as a function of $\Delta_E^{\mathrm{SNR}}$, $\Delta_E^{R_{\mathrm{min},1}}$,  $\Delta_E^{R_{\mathrm{min},2}}$, and $\Delta_E^{h}$. The figure shows that introducing a margin, $\Delta_E^{h}$, that accounts for potential deviations of the maximum height error does not  affect the performance. This is due to the fact that constraint $\mathrm{C9}$ is inactive for the system parameters considered. However, accounting for possible 10\% deviations of the minimum sensing data rate for the master and slave \acp{uav} comes with a 4.1\% and 5.4\% loss in total coverage, respectively. Additionally, the introduced \ac{snr} decorrelation margin, $\Delta_E^{\mathrm{SNR}}$, has a more significant effect on the performance, as a loss of 27.8\% in coverage is observed for a 10\% deviation. This illustrates the trade-off between the robustness of the obtained solution and the achievable coverage.
\section{Conclusion}
In this paper, we investigated joint formation and resource allocation optimization for bistatic \ac{uav}-based \ac{3d} {\ac{insar}} sensing. To assess the \ac{3d} sensing performance, {the} \ac{insar} coverage, interferometric coherence, \ac{hoa}, and relative height accuracy were introduced as relevant \ac{insar}-specific performance metrics. A non-convex non-smooth optimization problem was  formulated and solved for the maximization of the bistatic \ac{insar} ground coverage while enforcing energy, communication, and sensing performance constraints. {Our} simulation results confirmed the  superior performance of the proposed algorithm compared to {several}  benchmark schemes and emphasized the {relevance} of optimizing the velocity of both radar platforms and the look angle of the slave \ac{uav} for maximum \ac{insar} coverage.  We {revealed} the {importance of} low \ac{uav} speed for {high} \ac{insar} performance  and showed that large areas can be mapped and offloaded to ground in real time while ground object heights {can be} estimated with high precision. Lastly, we highlighted the significant impact of the \ac{uav}-\ac{gs} communication link on the \ac{uav} formation, {and} therefore, sensing performance.\newline {Potential topics for future work include the consideration of non-linear trajectories. Additionally, extending current results to the case of more than two \acp{uav} case is a promising research direction, for which, different communication schemes and sensing performance metrics can be used. Finally, considering a priori knowledge about the imaged scene in the optimization framework, such as the terrain type, moisture, or roughness, can also be explored.}}
\appendices
\ifarxiv
{\section{Proof of Proposition \ref{prop:equivalence}} \label{app:monotonic}
Constraint $\mathrm{C2}$ makes sub-problem $\mathrm{(P.1.b)}$ a one-dimensional optimization problem \ac{wrt} $z_1$, as it {specifies} a linear relation between {the} $x$- and $z$-coordinates of $U_1$. In fact, constraint $\mathrm{C2}$ imposes that $r_1(\mathbf{q}_1)=\frac{z_1}{\cos(\theta_1)}$, {and} therefore, constraints $\mathrm{C1}, \mathrm{C3}, \mathrm{C5}, $ and $\mathrm{C8}$ yield lower and upper bounds on $z_1$ that are collected  in variables $Z_{\rm max}$ and $Z_{\rm min}$, defined in (\ref{eq:Zmax}) and (\ref{eq:Zmin}), respectively. In particular, for constraint $\mathrm{C5}$, the lower and upper bounds on altitude $z_1$ are derived by  decomposing the baseline into its parallel and perpendicular components, i.e., $b(\mathbf{q}_1,\mathbf{q}_2)=\sqrt{b_{\bot}^2(\mathbf{q}_2)+b_{\parallel}^2(\mathbf{q}_1,\mathbf{q}_2)}$, where its parallel component is given by: 
\begin{equation} 
	b_{\parallel}(\mathbf{q}_1,\mathbf{q}_2)=
	b(\mathbf{q}_1,\mathbf{q}_2)  \sin\Big(\theta_1- \alpha(\mathbf{q}_1,\mathbf{q}_2)\Big).
\end{equation} 
 Note that, based on constraint $\mathrm{C2}$, it can be shown that the perpendicular baseline only depends on $\mathbf{q}_2$, see Appendix A in \cite{conf2}.  Hence, the baseline can be rewritten as {follows}:
\begin{equation}\label{eq:app:baseline}
		b^2(\mathbf{q}_1,\mathbf{q}_2)=\big(r_1(\mathbf{q}_1) -r_2(\mathbf{q}_2)\cos(\theta_1-\theta_2(\mathbf{q}_2)) \big)^2+ b_{\bot}^2(\mathbf{q}_2).
\end{equation}
Here, we only focus on the case where  $b_{\mathrm{min}}> b_{\bot}(\mathbf{q}_2)$, as constraint $\mathrm{C5}$ is always satisfied otherwise. Therefore, based on (\ref{eq:app:baseline}), the bounds on $z_1$ are given by: 
\begin{align}
  \frac{z_1}{\cos(\theta_1)}=r_1(\mathbf{q}_1)  \geq r_2(\mathbf{q}_2)\cos(\theta_1-\theta_2(\mathbf{q}_2)) + \sqrt{b^2_{\mathrm{min}}-b^2_{\bot}},\\ 
 \frac{z_1}{\cos(\theta_1)}=r_1(\mathbf{q}_1)  \leq r_2(\mathbf{q}_2)\cos(\theta_1-\theta_2(\mathbf{q}_2))- \sqrt{b^2_{\mathrm{min}}-b^2_{\bot}}.
\end{align}
Lastly, constraint $\overline{\mathrm{C11}}$ is derived by expanding constraint $\mathrm{C11}$ and rearranging its terms. Terms that are increasing \ac{wrt} $z_1$ are { assigned to} function ${a_1}$, whereas terms that are decreasing \ac{wrt} $z_1$ are grouped in function ${a_2}$.
\else 
\fi}
{\section{Proof of Proposition \ref{prop:convergence}} \label{app:convergence}
\textbf{Algorithm} \ref{alg:ao} is an \ac{ao}-based algorithm that maximizes the total \ac{insar} coverage, $C_N$, and consists of four major steps. First, in sub-problem $\mathrm{(P.1.a)}$, the location of the slave \ac{uav}, $\mathbf{q}_2$, is optimized based on \textbf{Algorithm} \ref{alg:pso}. Second, in sub-problem $\mathrm{(P.1.b)}$, the location of the  master \ac{uav}, $\mathbf{q}_1$, is optimized based on \textbf{Algorithm} \ref{alg:polyblock}. Third, in sub-problem $\mathrm{(P.1.c)}$, the remaining of optimization variables, $\{\mathbf{v}_y, \mathbf{P}_{\rm com,1}, \mathbf{P}_{\rm com,2}\}$, are optimized based on \textbf{Algorithm} \ref{alg:sca}. Lastly, the velocity vector is adjusted using a step size $\psi$, as proposed in (\ref{eq:step_size}).  To prove the convergence of \textbf{Algorithm} \ref{alg:ao}, we show that the objective function is (i) non-decreasing in each iteration and (ii) upper-bounded \cite{notes_ao}.
\newline First, \textbf{Algorithm} \ref{alg:pso} does not guarantee, in general, non-decreasing objective values, however, adding particle $\mathbf{q}_2^{(m-1)}$ to the population in iteration $m$ ensures non-decreasing \ac{insar} coverage, as $\mathbf{q}_2^{(m-1)}$ provides the worst-case objective value. Second, the polyblock outer approximation algorithm is also non-decreasing as the current best value, $\mathrm{CBV}$, is only updated when a better solution is found. Third, the \ac{sca} \textbf{Algorithm} \ref{alg:sca} is also non-decreasing since for any generated solution, $\mathbf{v}_y$, the initial feasible solution, $\mathbf{v}_y^{(m-1)}$, represents the worst-case objective value, i.e., $\sum\limits_{n=1}^{N-1}v_y^{(m-1)}[n]\leq \sum\limits_{n=1}^{N-1}v_y[n]$. Since the \ac{uav} formation is fixed for \textbf{Algorithm} \ref{alg:sca}, then, the coverage $C_N$ is directly proportional to the summation of the velocities across all time slots (see (\ref{eq:insar_coverage})).  Therefore,  we conclude that  $C_N(\mathbf{q}_1^{(m)},\mathbf{q}_2^{(m)},\mathbf{v}_y^{(m-1)})\leq C_N(\mathbf{q}_1^{(m)},\mathbf{q}_2^{(m)},\mathbf{v}_y), \forall m$. Lastly, the velocity vector, adjusted using the step size $\psi \in [0, 1]$ and denoted by $\mathbf{v}_y^{(m)}$, lies element-wise between $\mathbf{v}_y^{(m-1)}$ and $\mathbf{v}_y$, therefore, the inequality $\sum\limits_{n=1}^{N-1}v_y^{(m-1)}[n]\leq \sum\limits_{n=1}^{N-1}v_y^{(m)}[n]\leq \sum\limits_{n=1}^{N-1}v_y[n]$ holds, $\forall m$. This results in $C_N(\mathbf{q}_1^{(m)},\mathbf{q}_2^{(m)},\mathbf{v}_y^{(m-1)})\leq C_N(\mathbf{q}_1^{(m)},\mathbf{q}_2^{(m)},\mathbf{v}_y^{(m)}), \forall m$, and overall, a non-decreasing objective function for \textbf{Algorithm} \ref{alg:ao}. \newline
To prove (ii), we note that the objective function, defined by the common footprint width of the platforms, is naturally upper-bounded by the individual antenna footprint widths of both master and slave \acp{uav}. Thus, we can show that $ C_N(\mathbf{q}_1,\mathbf{q}_2,\mathbf{v}_y)\leq \big( S_{\rm far}(\mathbf{q}_{\rm max}) - S_{\rm near}(\mathbf{q}_{\rm max})\big) (N-1)v_{\rm max}\delta_t$ , where $\mathbf{q}_{\rm max}=(x_t-z_{\rm max}\tan(\theta_1),z_{\rm max})^T$, which concludes the proof.}
\vspace{-3mm}
\bibliographystyle{IEEEtran}
\bibliography{biblio}
\end{document}